\documentclass[sigplan,screen]{acmart}%
\acmPrice{15.00}
\acmDOI{10.1145/3453483.3454029}
\acmYear{2021}
\copyrightyear{2021}
\acmSubmissionID{pldi21main-p27-p}
\acmISBN{978-1-4503-8391-2/21/06}
\acmConference[PLDI '21]{Proceedings of the 42nd ACM SIGPLAN International Conference on Programming Language Design and Implementation}{June 20--25, 2021}{Virtual, Canada}
\acmBooktitle{Proceedings of the 42nd ACM SIGPLAN International Conference on Programming Language Design and Implementation (PLDI '21), June 20--25, 2021, Virtual, Canada}
\setcopyright{none}

\usepackage{booktabs}   %
\usepackage{subcaption} %
\usepackage{wrapfig}
\usepackage{url}                  %
\usepackage{listings}          %
\usepackage{enumitem}
\usepackage[font=small,labelfont=bf]{caption}
\usepackage{tikz}%
\usetikzlibrary{fit,shapes}
\pgfdeclarelayer{bg}    %
\pgfsetlayers{bg,main}  %
\usepackage{amsmath}
\usepackage{bm}
\usepackage{tabu}
\usetikzlibrary{positioning,shapes.geometric}
\usepgflibrary{decorations.shapes}
\usepackage{comment}
\usepackage[capitalise]{cleveref}
\usepackage{natbib}
\usepackage{qcircuit}
\usepackage{braket}
\usetikzlibrary{fit,decorations}

\ccsdesc[500]{Theory of computation~Quantum information theory}
\ccsdesc[500]{Theory of computation~Program analysis}

\usepackage{mathtools}

\usepackage{graphicx}
\usepackage{subcaption}
\usepackage{appendix}

\usepackage{mathpartir}

\newcommand{\change}[1]{{#1}}
\newcommand{\fix}[1]{{#1}}

\newcommand{\noise}{\omega}

\newcommand{\para}[1]{\vspace{9pt}\noindent{\textbf{#1. }}}

\newcommand{\framework}{Gleipnir}

\begin{document}

\title{\framework{}: Toward Practical Error Analysis for Quantum Programs (Extended Version)}

\author{Runzhou Tao}
\affiliation{
  \institution{Columbia University}            %
  \city{New York}
  \state{NY}
  \country{USA}                    %
}
\email{runzhou.tao@columbia.edu}          %

\author{Yunong Shi}
\authornote{Now affiliated with Amazon.}                                        %
\affiliation{
  \institution{The University of Chicago}            %
  \city{Chicago}
  \state{IL}
  \country{USA}                    %
}
\email{yunong@uchicago.edu}          %

\author{Jianan Yao}
\affiliation{
  \institution{Columbia University}            %
  \city{New York}
  \state{NY}
  \country{USA}                    %
}
\email{jy3022@columbia.edu}          %
\author{John Hui}
\affiliation{
  \institution{Columbia University}            %
  \city{New York}
  \state{NY}
  \country{USA}                    %
}
\email{j-hui@cs.columbia.edu}          %

\author{Frederic T. Chong}
\authornote{F. Chong is also Chief Scientist at Super.tech and an advisor to Quantum Circuits, Inc.}
\affiliation{
  \institution{The University of Chicago}            %
}
\affiliation{
  \institution{Super.tech}            %
  \city{Chicago}
  \state{IL}
  \country{USA}                    %
}
\email{chong@cs.uchicago.edu}          %

\author{Ronghui Gu}
\affiliation{
  \institution{Columbia University}            %
}
\affiliation{
  \city{New York}
  \state{NY}
  \country{USA}                    %
}
\email{ronghui.gu@columbia.edu}          %

\begin{abstract}
    Practical
    error analysis is essential for the design, optimization, 
    and evaluation of {Noisy Intermediate-Scale Quantum} (NISQ) computing. 
    However, bounding errors in quantum programs is a grand challenge,
    because the effects of quantum errors depend on
    exponentially large quantum states.
    In this work, we present \framework{}, a novel methodology
    toward practically computing verified error bounds in quantum programs.
    \framework{} introduces the $(\hat\rho,\delta)$-diamond norm,
    an error metric constrained by a quantum predicate consisting of 
    the approximate state $\hat\rho$ and its distance $\delta$ to the ideal state $\rho$.
    This predicate $(\hat\rho,\delta)$ can be computed {adaptively}
    using tensor networks based on Matrix Product States.
    \framework{} features a lightweight logic for reasoning
    about error bounds in noisy quantum programs,
    based on the $(\hat\rho,\delta)$-diamond norm metric.
    Our experimental results show that \framework{} is able to efficiently 
    generate tight error bounds for real-world quantum programs with 10 to 100 qubits,
    and can be used to 
    evaluate the error mitigation performance of
    quantum compiler transformations.

\end{abstract}

\keywords{quantum programming, error analysis, approximate computing}  %

\maketitle

\section{Introduction}

Recent quantum supremacy experiments~\cite{supremacy_2019}
have heralded
the {Noisy Intermediate-Scale Quantum} (NISQ) era~\cite{Preskill2018},
where %
noisy quantum computers with 50-100 qubits are used to achieve
tangible performance gains over classical computers.
While this goal is promising, there remains the engineering challenge of
accounting for erroneous quantum operations on noisy hardware~\cite{7927104}.
Compared to classical bits,
quantum bits (qubits) are much more fragile and error-prone.
The theory of Quantum Error Correction 
(QEC)~\cite{Devitt_2013,gottesman2009introduction,NielsenChuang,preskill1998lecture,Preskill_1998}
enables fault tolerant
computation~\cite{gottesman2009introduction,Campbell2017,preskill1997faulttolerant}
using redundant qubits,
but full fault tolerance is still prohibitively expensive for modern noisy devices---some
$10^3$ to $10^4$ physical qubits are required to encode
a single logical qubit~\cite{Knill_2005,Fowler_2012}.

To reconcile quantum computation with
\change{NISQ} computers,
quantum compilers perform transformations for
error mitigation~\cite{Wallman_2016}
and noise-adaptive optimization~\cite{murali2019noiseadaptive}.
To evaluate these compiler transformations,
we must compare the error bounds
of the source and compiled quantum programs.

Analyzing the error of quantum programs, however, is practically challenging.
Although one can naively calculate the ``distance'' (i.e., error) between the ideal and noisy outputs
using their matrix representations~\cite{NielsenChuang},
this approach is impractical for real-world quantum programs,
whose %
matrix representations %
can be exponentially large---for example, a 20-qubit quantum circuit is represented by
a $2^{20} \times 2^{20}$ matrix---too large to feasibly compute.

Rather than directly calculating the output error using matrix representations,
an alternative approach employs \emph{error metrics}, which can be computed more efficiently.
A common error metric for quantum programs is
the unconstrained diamond norm~\cite{aharonov1998quantum}.
However, this metric merely gives a {\it worst-case} error analysis:
it is calculated only using quantum gates' noise models,
without taking into account any information about the quantum state.
In extreme cases, 
it overestimates errors by up to two orders of magnitude~\cite{Wallman_2014}.
A more realistic metric must %
take the input quantum state into account,
since this also affects the output error.

The logic of quantum robustness (LQR)~\cite{hung2019} %
incorporates quantum states in the error metrics to compute tighter error bounds.
This work introduces the $(Q,\lambda)$-diamond norm,
which analyzes the output error given that
the input quantum state satisfies some
quantum predicate $Q$ to degree $\lambda$.
LQR
extends the Quantum Hoare Logic~\cite{ying2016foundations}
with the $(Q,\lambda)$-diamond norm
to produce
logical judgments of the form $(Q, \lambda)\vdash\widetilde{P}\leq\epsilon$,
which deduces the error bound $\epsilon$ for a noisy program $\widetilde{P}$.
While theoretically promising, this work raises open questions in practice.
Consider the following sequence rule in LQR: %
\[
\inferrule{
(Q_1, \lambda)\vdash\widetilde{P_1}\leq\epsilon_1
\and \{Q_1\}P_1\{Q_2\}
\and (Q_2, \lambda)\vdash\widetilde{P_2}\leq\epsilon_2
}{(Q_1, \lambda)\vdash(\widetilde{P_1};\widetilde{P_2})\leq\epsilon_1+\epsilon_2}.
\]
It is unclear how to obtain a
quantum predicate $Q_2$
that is  a valid postcondition after executing $P_1$
while being strong enough to produce useful error bounds for $\widetilde{P_2}$.

This paper presents \framework{},
an adaptive error analysis %
methodology for quantum programs that addresses
the above practical challenges and
answers the following three open questions:
(1) How to compute suitable constraints for the input quantum state used by the error metrics?
(2) How to reason about error bounds without manually verifying quantum programs
with respect to pre- and postconditions?
(3) How practical is it to
compute verified error bounds for quantum programs
and evaluate the error mitigation performance of quantum compiler transformations? %

First, 
in prior work,
seaching for
a non-trivial postcondition $(Q,\lambda)$ for a given quantum program
is prohibitively costly:
existing methods either compute  postconditions
by fully simulating quantum programs using 
matrix representations~\cite{ying2016foundations}, %
or reduce this problem to an SDP 
(Semi-Definite Programming) problem whose size is exponential to
the number of qubits used in the quantum program~\cite{ying2017invariants}.
In practice,
for large quantum programs ($\geq$ 20 qubits),
these methods %
cannot produce any postconditions other than
$(I, 0)$
(i.e., the identity matrix $I$ to degree 0, %
analogous to a ``true'' predicate),
reducing the $(Q,\lambda)$-diamond norm to the unconstrained diamond norm
and failing to yield non-trivial error bounds.

To overcome this limitation,
\framework{} introduces the  $(\hat{\rho}, \delta)$-diamond norm,
a new error metric 
for input quantum states whose
distance 
from some \emph{approximated} quantum state $\hat{\rho}$
is bounded by $\delta$.
Given a quantum program and a predicate $(\hat{\rho}, \delta)$,
\framework{} computes its diamond norm by
reducing it to a constant size SDP problem.

To obtain the predicate $(\hat{\rho}, \delta)$,
\framework{} uses Matrix Product State (MPS) tensor networks \cite{mps} %
to represent and approximate quantum states.
Rather than fully simulating the quantum program or
producing an exponentially complex SDP problem,
our MPS-based approach computes a tensor network $TN(\rho_0, P)$ that
\emph{approximates} $(\hat{\rho}, \delta)$
for some input state $\rho_0$ and program $P$.
By \emph{dropping} insignificant singular values when exceeding
the given MPS size
during the approximation, $TN(\rho_0, P)$ can be computed in  polynomial time
with respect to the size of the MPS tensor network,
the number of qubits, and the number of quantum gates. %
In contrast with prior work,
our MPS-based approach is \emph{adaptive}---one may adjust the approximation precision by varying the size of the MPS
such that %
tighter error bounds can be computed using greater computational 
resources.
\framework{} provides more flexibility between
the tight but inefficient full simulation
and the efficient but unrealistic worst-case analysis.

Second, instead of verifying a  predicate using Quantum Hoare Logic,
\framework{} develops a %
lightweight logic 
based on $(\hat{\rho}, \delta)$-diamond norms
for reasoning about quantum program error, using judgments of the form:
\[
    (\hat{\rho}, \delta)\vdash\widetilde{P}_\noise\leq\epsilon.
\]
This judgement states that the error of the noisy program $\widetilde{P}_\noise$ 
under the  noise model  $\noise$
is upper-bounded by $\epsilon$ 
when the input state is constrained by
$(\hat{\rho}, \delta)$.
As shown in the sequence rule of our quantum error logic:
\begin{small}
\[
\inferrule{
(\hat{\rho}, \delta)\vdash \widetilde{P}_{1\noise}\leq \epsilon_1
\quad\ TN(\hat{\rho}, P_1) = (\hat{\rho}', \delta')
\quad\
(\hat{\rho}', \delta + \delta')\vdash\widetilde{P}_{2\noise}\leq\epsilon_2
}{(\hat{\rho}, \delta)\vdash\ \widetilde{P}_{1\noise};\widetilde{P}_{2\noise}\ \leq\epsilon_1 + \epsilon_2},
\]
\end{small}%
the approximated state $\hat{\rho}'$ and its distance $\delta'$
are  \emph{computed} using the MPS tensor network $TN$.

our sequence rule eliminates the cost of searching for and
validating non-trivial postconditions by directly computing $(\hat{\rho}, \delta)$.
We prove the correctness of $TN$, which ensures that
the resulting state of executing $P_1$ satisfies
the  predicate $(\hat{\rho}', \delta + \delta')$.

Third, we enable the practical error analysis of quantum programs and transformations,
which was previously only theoretically possible but infeasible 
due to the limitations of prior work.
To understand the scalability and limitation of our error analysis methodology, 
we conducted case studies using %
two classes of quantum programs that are expected to be most useful in the near-term---the Quantum Approximate Optimization Algorithm~\cite{farhi2014quantum}
and the Ising model~\cite{google2020hartree}---with qubits ranging from 10 to 100.
Our measurements show that, with 128-wide MPS networks,
\framework{} can always generate error bounds within 6 minutes.
For small programs ($\leq$ 10 qubits), \framework{}'s error bounds
\change{are} as precise as the ones generated using full simulation.
For  large programs ($\geq$ 20 qubits),
\framework{}'s error bounds
are $15\%$ to $30\%$ tighter than
those calculated using unconstrained diamond norms,
while full simulation invariably times out after
24 hours.

We %
explored \framework{}'s effectiveness in  evaluating
the error mitigation performance of quantum compiler transformations.
We conducted a case study evaluating
qubit mapping protocols~\cite{murali2019noiseadaptive}
and showed that
the %
ranking for different transformations using
the error bounds generated by our methodology
is consistent with the ranking using errors measured from the real-world experimental data.

Throughout this paper, we
address the key practical limitations of error analysis for quantum programs.
In summary, our main contributions  are:

\begin{itemize}
  \setlength\itemsep{0.5em}
\item The $(\hat{\rho}, \delta)$-diamond norm,
    a new error metric constrained by the input quantum state
    that can be efficiently computed using constant-size SDPs.
\item An MPS tensor network approach
    to adaptively compute the quatum predicate $(\hat{\rho}, \delta)$.
    
\item A lightweight logic for 
reasoning about quantum error bounds
without the need to verify quantum predicates.
\item Case studies using quantum programs and transformations on real quantum devices,
    demonstrating the feasability of adaptive quantum error analysis for
    computing verified error bounds for quantum programs and
    evaluating the error mitigation \fix{performance} of quantum compilation.
\end{itemize}

\section{Quantum Programming Background}

This section introduces  basic notations and terminologies for quantum programming 
that will be used throughout the paper.
\fix{Please refer to \citet{r3} for more detailed background knowledge.}

\vspace{-5pt}
\para{Notation}
In this paper, we use Dirac notation, or ``bra-ket'' notation, to represent quantum states.
The ``ket'' notation $\ket{\psi}$ denotes a column vector, which corresponds to a pure quantum state.
The ``bra'' notation $\bra{\psi}$ denotes its conjugate transpose, a row vector.
$\braket{\phi|\psi}$ represents the inner product of two vectors,
and $\ket{\psi}\bra{\phi}$ the outer product. 
We use $\rho$ to denote a density matrix \fix{(defined in Section~\ref{sec:basic})}, a matrix that represents a mixed quantum state.
$U$ usually denotes a unitary matrix which represents quantum gates,
while $U^\dag$ denotes its conjugate transpose.
Curly letters such as $\mathcal{U}$ denote noisy or ideal quantum operations,
represented by maps between density matrices (superoperators).
Upper case Greek letters such as $\Phi$ represent quantum noise as superoperators.

\subsection{Quantum computing basics}
\label{sec:basic}
\vspace{-5pt}
\para{Quantum states}
The simplest quantum state is a quantum bit---a \emph{qubit}.
Unlike a classical bit, a qubit's state can be the superposition of two logical states,
$\ket{0}$ and $\ket{1}$, that correspond to classical logical states $0$ and $1$.
In general, a qubit is a unit vector in the 2-dimensional Hilbert space $\mathbb{C}^2$,
with $\ket{0} := [1, 0]^\dag$ and $\ket{1} := [0,1]^\dag$.
In Dirac's notation, we represent a qubit as $\ket{\psi} = \alpha\ket{0} + \beta\ket{1}$,
where $|\alpha|^2 + |\beta|^2 = 1$. 

Generally speaking, the state of a quantum program may comprise many qubits.
An $n$-qubit state can be represented by a unit vector
in $2^n$-dimensional Hilbert space $\mathbb{C}^{2^n}$.
For example, a 3-qubit state can be described by an $8$-dimensional complex vector,
which captures a superposition of 8 basis states,
$\ket{000}$, $\ket{001}$, $\ket{010}$, $\ldots$, $\ket{111}$.

Besides the pure quantum states described above,
there are also \emph{classically mixed} quantum states, i.e., noisy states.
An $n$-qubit mixed state can be represented by a $2^n \times 2^n$ \emph{density matrix}
$\rho = \sum_i p_i \ket{\phi_i} \bra{\phi_i}$,
which states that the state has $p_i$ probability to be $\ket{\phi_i}$.
For example, a mixed state with half probability of $\ket{0}$ and $\ket{1}$
can be represented by $\frac{\ket{0}\bra{0}+\ket{1}\bra{1}}{2} = I/2$,
where $I$ is the identity matrix.

\para{Quantum gates}
Quantum states are manipulated by the application of \emph{quantum gates},
described by unitary matrix representations~\cite{NielsenChuang}.
Figure~\ref{fig:gatematrix} shows the matrix representations of 
some common  gates.
Applying an operator $U$ to a quantum state $\ket{\phi}$ results in the state $U\ket{\phi}$,
and applying it to a density matrix $\rho = \sum_i p_i \ket{\phi_i} \bra{\phi_i}$ gives $U \rho U^{\dag}$.
For example, the bit flip gate $X$ maps $\ket{0}$ to $\ket{1}$ and $\ket{1}$ to $\ket{0}$,
while the Hadamard gate $H$ maps $\ket{0}$ to $\frac{\ket{0} + \ket{1}}{\sqrt{2}}$.
There are also multi-qubit gates, such as $CNOT$,
which does not change $\ket{00}$ and $\ket{01}$
but maps $\ket{10}$ and $\ket{11}$ to each other.
Applying a gate on a subset of qubits will not change other qubits.
For example, applying the $X$ gate to the first qubit of $\frac{\ket{00} + \ket{11}}{\sqrt{2}}$
will result in  $\frac{\ket{10} + \ket{01}}{\sqrt{2}}$.
This can be seen as an extension $X \otimes I$ of the matrix to a larger space
using a tensor product. %

\begin{figure}[t]
    \centering \small
        $ X = \begin{bmatrix} 0 & 1 \\ 1 & 0 \end{bmatrix},\
        Z = \begin{bmatrix} 1 & 0 \\ 0 & -1 \end{bmatrix},\
        H =\frac{1}{\sqrt{2}}\begin{bmatrix} 1 & 1 \\ 1 & -1 \end{bmatrix}$,\
        $ CNOT = $ \scriptsize\setlength\arraycolsep{2pt}
                $\begin{bmatrix} 1 & 0 & 0 & 0 \\ 0 & 1 & 0 & 0 \\ 0 & 0 & 0 & 1 \\ 0 & 0 & 1 & 0 \end{bmatrix}$
    \caption{Matrix representations of common quantum gates.
    $X$ denotes a bit flip, $Z$ denotes a phase flip,  $H$ denotes a Hadamard gate,
    and $CNOT$ denotes a controlled NOT gate.
    }
    \label{fig:gatematrix}
\end{figure}

\para{Quantum measurements}
Measurements extract classical information from quantum states
and collapse the quantum state according to \emph{projection matrices}
$M_0$ and $M_1$.
When we measure some state $\rho$,
we will obtain the result $0$ with collapsed state
$M_0 \rho M_0^\dag / p_0$ and probability $p_0 = \mathrm{tr} (M_0 \rho M_0^\dag)$,
or the result $1$ with collapsed state $M_1 \rho M_1^\dag / p_1$
and probability $p_1 = \mathrm{tr} (M_1 \rho M_1^\dag)$.

Both quantum gates and quantum measurements act linearly on density matrices
and can be expressed as \emph{superoperators}, which are
completely positive trace-preserving maps
$\mathcal{E} \in L(\mathcal{H}) : \mathcal{H}_n\rightarrow\mathcal{H}_m$,
where $\mathcal{H}_{n}$ is the density matrix space of dimension $n$
and $L$ is the space of linear operators.

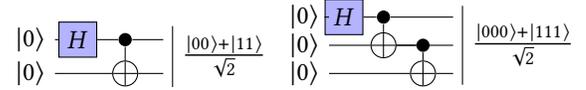
\begin{figure}[t]
\begin{subfigure}[b]{.2\textwidth}
\centering
    \begin{tikzpicture}[scale=0.9]
    \tikzstyle{operator} = [draw,fill=blue!30] 
    \tikzstyle{phase} = [fill,shape=circle,minimum size=5pt,inner sep=0pt]
    \tikzstyle{surround} = [fill=blue!10,thick,draw=black,rounded corners=2mm]
    \tikzset{XOR/.style={draw,circle,append after command={
        [shorten >=\pgflinewidth, shorten <=\pgflinewidth,]
        (\tikzlastnode.north) edge (\tikzlastnode.south)
        (\tikzlastnode.east) edge (\tikzlastnode.west)
    } } }
    \node at (0,0) (q1) {$\ket{0}$};
    \node at (0,-0.5) (q2) {$\ket{0}$};
    \node[operator] (op11) at (0.7,0) {$H$} edge [-] (q1);
    \node[phase] (phase11) at (1.4,0) {} edge [-] (op11);
    \node[XOR] (phase12) at (1.4,-0.5) {} edge [-] (q2);
    \draw[-] (phase11) -- (phase12);
    \node (end1) at (2.1,0) {} edge [-] (phase11);
    \node (end2) at (2.1,-0.5) {} edge [-] (phase12);
    \draw (2.1,0.2) to
	node[midway,right] (bracket) {$\frac{\ket{00}+\ket{11}}{\sqrt{2}}$}
	(2.1,-0.7);
    \end{tikzpicture}
\end{subfigure}
\begin{subfigure}[b]{.2\textwidth}
\centering
    \begin{tikzpicture}[scale=0.75]
        \tikzstyle{operator} = [draw,fill=blue!30] 
        \tikzstyle{phase} = [fill,shape=circle,minimum size=5pt,inner sep=0pt]
        \tikzstyle{surround} = [fill=blue!10,thick,draw=black,rounded corners=2mm]
        \tikzset{XOR/.style={draw,circle,append after command={
            [shorten >=\pgflinewidth, shorten <=\pgflinewidth,]
            (\tikzlastnode.north) edge (\tikzlastnode.south)
            (\tikzlastnode.east) edge (\tikzlastnode.west)
            }
            
        }
    }
        \node at (0,0) (q1) {$\ket{0}$};
        \node at (0,-0.5) (q2) {$\ket{0}$};
        \node at (0,-1) (q3) {$\ket{0}$};
        \node[operator] (op11) at (0.7,0) {$H$} edge [-] (q1);
        \node[phase] (phase11) at (1.4,0) {} edge [-] (op11);
        \node[XOR] (phase12) at (1.4,-0.5) {} edge [-] (q2);
        \draw[-] (phase11) -- (phase12);
        \node[phase] (phase22) at (2.1,-0.5) {} edge [-] (phase12);
        \node[XOR] (phase23) at (2.1,-1) {} edge [-] (q3);
        \draw[-] (phase22) -- (phase23);
        \node (end1) at (2.8,0) {} edge [-] (phase11);
        \node (end2) at (2.8,-0.5) {} edge [-] (phase12);
        \node (end3) at (2.8,-1) {} edge [-] (phase23);
        \draw (2.8,0.2) to
    	node[midway,right] (bracket) {$\frac{\ket{000}+\ket{111}}{\sqrt{2}}$}
    	(2.8,-1.2);
    \end{tikzpicture}
\end{subfigure}

\caption{
Two quantum circuits,
producing the 2-qubit (left) and 3-qubit (right) GHZ states.
}
\label{fig:ghz}
\end{figure}

\subsection{Quantum programs}
\label{sec:qprog}
Quantum programs comprise a configuration of quantum gates and measurements,
called a quantum circuit.
Graphically, qubits are represented as wires, and gates as boxes joining the wires;
CNOT gates are represented by a dot on the control qubit linked
with an $\oplus$ on the other qubit.

\begin{example}[GHZ circuit]
\label{ex:ghz}
The Greenberger–Horne–\\Zeilinger (GHZ) state~\cite{greenberger1989going} is
a class of entangled quantum states
used in many quantum communication protocols~\cite{hillery1999quantum}.
The simplest GHZ state is the 2-qubit GHZ state,
which is $\frac{\ket{00} + \ket{11}}{\sqrt{2}}$ in Dirac notation.
Figure \ref{fig:ghz} shows a typical graphical representation of a quantum circuit
that produces the 2-qubit GHZ state.
\end{example}

\para{Syntax}
The %
syntax of quantum programs is as follows:
\begin{align*}
        P &::=\  \texttt{skip} \ \mid \ 
            P_1; P_2 \  \mid \ 
            U (q_1, \ldots, q_k) \ \\ & \mid \ 
            \texttt{if } q = \ket{0} \texttt{then } P_0 \texttt{ else } P_1.
\end{align*}

\noindent{}Each component behaves similarly to its classical counterpart:
\texttt{skip} denotes the empty program;
$P_1; P_2$ sequences programs;
$U (q_1, \ldots, q_k)$ applies the $k$-qubit gate $U$ to the qubits $q_1, \ldots, q_k$;
$\texttt{if } q = \ket{0} \texttt{ then } P_0 \texttt{ else } P_1$ measures the qubit $q$,
executes $P_0$ if the result is $0$, and executes $P_1$ otherwise.
The difference between classical and quantum programs is that
the measurement in the \texttt{if} statement will collapse the state,
and the branch is executed on the collapsed state.
Using this syntax, the 2-qubit GHZ state circuit in Figure~\ref{fig:ghz} is written as:
\[
    H(q_0); CNOT(q_0, q_1).
\]

\noindent{}Note that this
work  currently does not  consider advanced quantum program constructs such as quantum loops,
as these are not likely to be supported on near-term quantum machines.

\para{Denotational semantics}
The denotational semantics of quantum programs are defined as superoperators acting on density matrices $\rho$, shown in Figure~\ref{fig:semantics}.
An empty program keeps the state unchanged;
a sequence of operations are applied to the state one by one;
a single quantum gate is directly applied as a superoperator\footnote{\fix{The matrix $U$ in Figure~\ref{fig:semantics} denotes the gate matrix (like in Figure~\ref{fig:gatematrix}) extended with identity operator on unaffected qubits.}};
a measurement branch statement maps the state into a classical mix
of the two results from executing the two branches.

\begin{figure}[t]
\begin{align*}
    [\![ \texttt{skip} ]\!] (\rho) :=& \rho \\
    [\![ P_1; P_2 ]\!] (\rho) :=& [\![ P_2 ]\!] ([\![ P_1 ]\!] (\rho)) \\
    [\![ U (q_1, \ldots, q_k) ]\!] (\rho) :=& U\rho U^\dag \\
    [\![ \texttt{if } q = \ket{0} \texttt{then } P_0 \texttt{ else } P_1 ]\!] (\rho) :=& [\![ P_0 ]\!] (M_0 \rho M_0^\dag) \ + \\ & [\![ P_1 ]\!] (M_1 \rho M_1^\dag)
\end{align*}
\vspace{-10pt}
\caption{Denotational semantics of quantum programs.}
\label{fig:semantics}
\end{figure}

\subsection{Quantum errors}
\label{sec:qerr}
Quantum programs are always noisy, and that noise may (undesirably) perturb the quantum state.
For example, the bit flip noise flips the state of a qubit with probability $p$.
This noise can be represented by a superoperator $\Phi$ such that:
$$\Phi(\rho) = (1-p)\rho + pX\rho X,$$
i.e., the state remains the same with probability $1-p$
and changes to $X\rho X$ with probability $p$,
where $X$ is the matrix representation of the bit flip gate (see Figure~\ref{fig:gatematrix}).
Generally, all effects from quantum noise can be represented by superoperators.

\para{Noisy quantum programs}
The noise model $\noise$ specifies the noisy version $\widetilde{U}_\noise$ of each gate $U$ on the \fix{target} noisy device,
used to specify  noisy quantum programs $\widetilde{P}_\noise$.
The noisy semantics $[\![ P ]\!]_\noise$ of program $P$ can be defined
as the semantics $[\![ \widetilde{P}_\noise ]\!]$ of the noisy program
$\widetilde{P}_\noise$,
whose semantics are similar to that of a noiseless program.
The rules of skip, sequence, and measurement statements remain the same,
while for gate application, the noisy version of each gate is applied as follows:
\[
    [\![ U (q_1, \ldots, q_k) ]\!]_\noise (\rho) =
    [\![ \widetilde{U}_\noise (q_1, \ldots, q_k) ]\!] (\rho) =
    \widetilde{\mathcal{U}}_\noise(\rho),
\]
where $\widetilde{\mathcal{U}}_\noise$ is the superoperator representation of %
$\widetilde{U}_\noise$.

\para{Metrics for quantum errors}
To quantitatively evaluate the effect of noise,
we need to measure some notion of ``distance'' between quantum states.
The \emph{trace distance} $\| \rho_n - \rho_{\text{Id}} \|_1$ measures the distance between the noisy state $\rho_{n}$
and the ideal, noiseless state $\rho_{\text{Id}}$: 
\begin{align*}
   \|\rho_{n} - \rho_{\text{Id}}\|_1=\max_P \fix{\text{ tr}}( P(\rho_{n}-\rho_{\text{Id}})),
\end{align*}
\noindent{}where $P$ is a positive semidefinite matrix with trace $1$ and $\text{tr}$ denotes the trace of a matrix.
The trace distance can be seen as a special case of the Schatten-$p$ norm $||\cdot ||_{p}$, defined as:
\begin{align*}
    ||\cdot||_p := \Big(\text{tr}(\cdot^\dagger \cdot)^\frac{p}{2}\Big)^\frac{1}{p}.
\end{align*}
\noindent{}The trace distance measures the maximum statistical distance
over all possible measurements of two quantum states.
Note that trace distance cannot be directly calculated
\fix{without complete information of the two quantum states}.

The \emph{diamond norm} metric is typically used to obtain
a {\it worst case} error bound.
The diamond norm between two superoperators $\mathcal{U}$ and $\mathcal{E}$ is defined as:
\begin{align*}
    ||\mathcal{U}- \mathcal{E}||_{\diamond}
    \quad =& \quad
            \max_{\rho :\ \text{tr}(\rho)=1}
            \frac{1}{2}\|\mathcal{U} \otimes \mathcal{I}(\rho) -\mathcal{E} \otimes\mathcal{I}(\rho) \|_{1},
\end{align*}
\noindent{}\fix{where $\mathcal{I}$ is the identity superoperator over some auxiliary space.} Intuitively, this formula calculates the maximum trace distance between the output state
after applying the erroneous operation versus applying the noiseless operation,
for any {\it arbitrary} input state.
Diamond norms can be efficiently computed by
simple Semi-Definite Programs (SDP)~\cite{watrous2013simpler}. \fix{Please refer to \citet{sdp} for more  background on SDP}. 

However, as shown by the Wallman-Flammia bound \cite{Wallman_2014},
diamond norms may overestimate errors by up to two orders of magnitude,
precluding its application in more precise analyses of noisy quantum programs.
The diamond norm metric fails to incorporate information about the
quantum state of the circuit that may help tighten the error bound.
For example,
a bit flip error ($X$ gate) does nothing to the state $\frac{\sqrt{2}}{2}\big(\ket{0}+\ket{1}\big)$
(the state is unchanged after flipping $\ket{0}$ and $\ket{1}$),
but flips the $\ket{1}$ state to $\ket{0}$.
However, both trace distance and diamond norm are agnostic to the input state,
and thus limit our ability to %
tightly bound the errors of quantum circuits.

$(Q, \lambda)-$diamond norm \cite{hung2019} is a more fine-grained metric:
\begin{small}
\begin{align*}
    ||\mathcal{U}- \mathcal{E}||_{(Q, \lambda)}
     &:=
            \max_{\rho :\ \text{tr}(\rho)=1, \text{tr}(Q\rho) \ge \lambda}
            \frac{1}{2}\|\mathcal{U} \otimes \mathcal{I}(\rho) -\mathcal{E} \otimes\mathcal{I}(\rho) \|_1.
\end{align*}
\end{small}%
\noindent{}Unlike the %
unconstrained diamond norm, 
the $(Q, \lambda)-$diamond norm
constrains  the input state to satisfy the predicate $Q$\fix{, a positive semidefinite and trace-1 matrix,} to degree $\lambda$;
specifically, the input state $\rho$ must satisfy $\text{tr}(Q\rho) \ge \lambda$.
The $(Q, \lambda)-$ diamond norm may produce tighter error bounds than
the unconstrained diamond norm by utilizing quantum state information,
but leaves open the problem of
practically computing a non-trivial predicate $Q$.

\section{\framework{} Workflow}
\label{sec:framework}

\tikzset{
    punkt/.style={
           rectangle,
           rounded corners,
           draw=black, very thick,
           text width=5.em,
           minimum height=2em,
           text centered},
    pil/.style={
           ->,
           thick,
           shorten <=2pt,
           shorten >=2pt,}
}

\begin{figure}
    \centering
    \includegraphics[width=\linewidth]{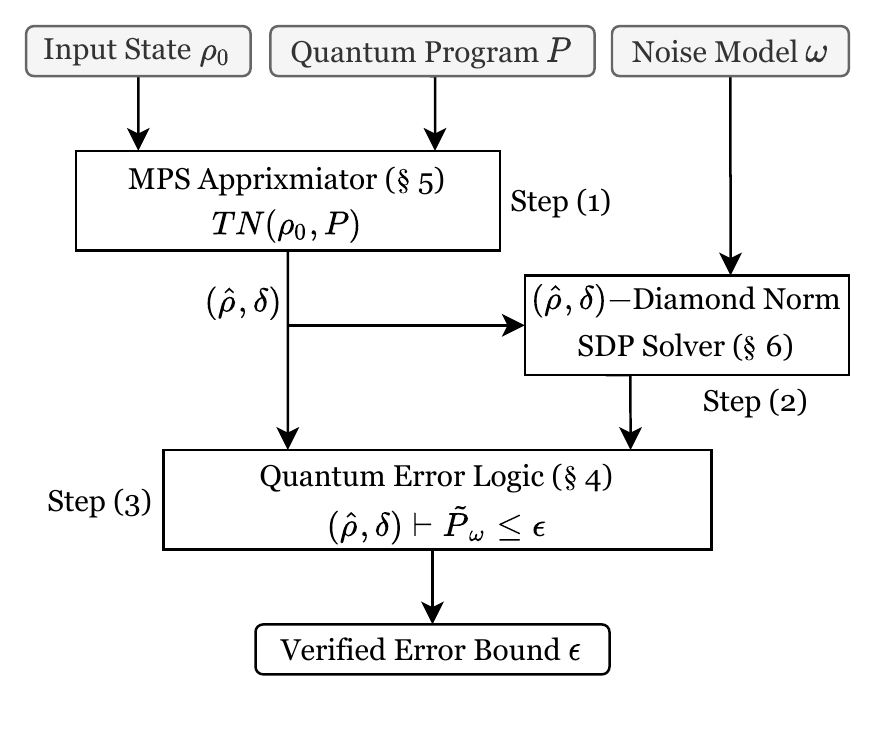}
    \vspace{-15pt}
    \caption{\framework{} workflow. %
    }
    \label{fig:workflow}
\end{figure}

\begin{figure*}[t]
\begin{mathpar}
\hspace{-20pt}
\inferrule*[Right=\textsc{Skip}]
{ }
{(\hat{\rho}, \delta) \vdash \widetilde{P}_\noise \leq 0}
\and
\inferrule*[Right=\textsc{Gate}]
{\|\widetilde{\mathcal{U}}_\noise - \mathcal{U} \|_{(\hat\rho, \delta)} \le \epsilon}
{(\hat{\rho}, \delta)\vdash\ \widetilde{U}_\noise(q_1, \ldots)\ \leq\epsilon}
\and
\inferrule* [Right=\textsc{Seq}]
{(\hat{\rho}, \delta)\vdash \widetilde{P}_{1\noise}\leq \epsilon_1
\quad
{TN}(\hat{\rho}, P_1) = (\hat{\rho}', \delta')
\quad 
(\hat{\rho}', \delta + \delta')\vdash\widetilde{P}_{2\noise}\leq\epsilon_2}
{(\hat{\rho}, \delta)\vdash\ \widetilde{P}_{1\noise};\widetilde{P}_{2\noise}\ \leq\epsilon_1 + \epsilon_2}
\and
\inferrule*[Right=\textsc{Weaken}]
{(\hat{\rho}, \delta')\vdash\widetilde{P}_\noise\leq\epsilon'
\\
\epsilon' \le \epsilon
\\ \delta' \ge \delta}
{(\hat{\rho}, \delta)\vdash\widetilde{P}_\noise\leq\epsilon} 
\and\qquad
\inferrule*[Right=\textsc{Meas}]
{(\hat{\rho}_0, \delta)\vdash\widetilde{P}_{0\noise}\leq\epsilon
\\
(\hat{\rho}_1, \delta)\vdash\widetilde{P}_{1\noise}\leq\epsilon}
{(\hat{\rho}, \delta)\vdash\ \Big(\texttt{if } q = \ket{0} \texttt{then } \widetilde{P}_{0\noise} \texttt{ else } \widetilde{P}_{1\noise} \Big)\ \leq(1-\delta)\epsilon + \delta}
\end{mathpar}

    \caption{Inference rules of the quantum error logic.}
    \vspace{10pt}
    \label{fig:rules}
\end{figure*}

To use the input
quantum state to tighten the computed error bound,
\framework{} introduces a new constrained diamond norm,
$(\hat{\rho}, \delta)$-diamond norm,
and a judgment $(\hat{\rho}, \delta)\vdash\widetilde{P}_\noise\leq\epsilon$
to reason about the error of quantum circuits.
Gleipnir  uses  Matrix Product State (MPS) tensor networks to approximate the quantum state
and compute the predicate $(\hat{\rho}, \delta)$.

Figure~\ref{fig:workflow} illustrates \framework{}'s workflow
for reasoning about the error bound of some quantum program $P$
with input state $\rho_0$ and noise model $\noise$ of quantum gates on the target device:

\begin{enumerate}
  \setlength\itemsep{0.75em}
\item[(1)] \framework{} first approximates the quantum state $\hat{\rho}$
and \fix{a sound overapproximation of} its distance $\delta$ to the ideal state $\rho$
using MPS tensor networks $TN(\rho_0, P) = (\hat{\rho}, \delta)$ %
(see Section~\ref{sec:tensor}). %

\item[(2)]  \framework{} then uses the constrained 
$(\hat{\rho}, \delta)$-diamond norm metric
to bound errors of noisy quantum gates
given a noise model $\noise$ of the target device.
\framework{} converts the problem of efficiently computing
the $(\hat{\rho}, \delta)$-diamond norm
to solving a polynomial-size SDP problem, %
given %
$(\hat{\rho}, \delta)$ computed
in  Step (1)  (see Section~\ref{sec:constrainedSDP}).

\item[(3)] \framework{} employs a lightweight quantum error logic to
compute the error bound of $\widetilde{P}_\noise$ using the 
 predicate $(\hat{\rho}, \delta)$ computed  in Step (1)
and the error bounds for all used quantum gates generated
by the SDP solver in Step (2)
(see Section~\ref{sec:logic}).
\end{enumerate}

Throughout this paper, we will return to the GHZ state circuit (\cref{ex:ghz})
as our running example.
This example uses 
the program $H(q_0); CNOT(q_0, q_1)$,
the input state $\ket{00} \bra{00}$,
and the noise model $\noise$,
describing the noisy gates $\widetilde{H}_\noise$ and $\widetilde{CNOT}_\noise$. 
Following the steps described above, we will use \framework{} to obtain
the final judgment of:
\[
    (\ket{00}\bra{00},\ 0)
    \vdash \Big(\widetilde{H}_\noise(q_0)\ ;\ \widetilde{CNOT}_\noise(q_0, q_1) \Big)
    \leq \epsilon,
\]
where $\epsilon$ is the total error bound of the noisy program.

\section{Quantum Error Logic}
\label{sec:logic}

We first introduce our lightweight logic
for reasoning about the error bounds of  quantum programs. %
In this section, we treat MPS tensor networks and
the algorithm to compute the $(\hat{\rho}, \delta)$-diamond norm as black boxes,
deferring their discussion to 
Sections~\ref{sec:tensor} and \ref{sec:constrainedSDP}, respectively.

The $(\hat{\rho},\delta)$-diamond norm is
defined as follows:
\begin{align*}
    \|\mathcal{U}- \mathcal{E}\|_{(\hat{\rho}, \delta)}
             &:=  \max_{\substack{\rho :\ \text{tr}(\rho)\ =\ 1,\\ \|\rho - \hat{\rho}\|_1 \le \delta}}
             \frac{1}{2}{\Big\|\mathcal{U}\otimes \mathcal{I}(\rho) - \mathcal{E} \otimes\mathcal{I}(\rho)\Big\|_1}.
\end{align*}
\noindent{}%
That is a diamond norm with the additional constraint that
the ideal input density matrix of $\rho$ needs to be within distance $\delta$
of $\hat{\rho}$, i.e., $T( \rho, \hat{\rho}) \le \delta$.

We use the judgment $(\hat{\rho}, \delta)\vdash\widetilde{P}_\noise\leq\epsilon$ 
to convey that when running the noisy program
$\widetilde{P}_\noise$ on an input state 
whose trace distance is at most $\delta$ from $\hat{\rho}$,
the trace distance between the noisy and noiseless outputs of program $P$
is at most $\epsilon$ under the noise model $\noise$ of the underlying device.

Figure~\ref{fig:rules} presents the five inference rules for our quantum error logic.
The \textsc{Skip} rule states that an empty program does not produce any noise. 
The \textsc{Gate} rule states that we can bound the error of a gate step
by calculating the gate's $(\hat{\rho}, \delta)$-diamond norm under the noise model $\noise$.
The \textsc{Weaken} rule states that the same error bound holds when we strengthen the precondition with a smaller approximation
bound $\delta'$.
The \textsc{Seq} rule states that the errors of a sequence can be summed together
with the help of the tensor network approximator ${TN}$.
The \textsc{Meas} rule bounds the error in an \texttt{if} statement,
with $\delta$ probability that the result of measuring the noisy input
differs from measuring state $\rho$, causing the wrong branch to \fix{be} executed.
Otherwise, the probability that the correct branch is executed is $1 - \delta$.
\fix{Given that in both branches, the error is bounded by a uniform value $\epsilon$,}
we multiply this probability by the error incurred in the branch,
and add it to the probability of taking the incorrect branch,
to obtain the error incurred by executing a quantum conditional statement.
\fix{The precondition in each branch is defined as $\hat{\rho}_0 = M_0 \hat{\rho} M_0^\dag / \text{tr}( M_0 \hat{\rho} M_0^\dag)$ and $\hat{\rho}_1 = M_1 \hat{\rho} M_1^\dag / \text{tr}( M_1 \hat{\rho} M_1^\dag)$.}

Our error logic contains two external components:
(1) $TN(\rho, P) = (\hat\rho, \delta)$, the tensor network approximator
used to approximate $[\![P]\!](\rho)$, obtaining $\hat\rho$ and an approximation error bound $\delta$; and (2)
$\| \cdot\|_{(\hat\rho, \delta)}$, the $(\hat{\rho}, \delta)$-diamond norm that characterizes the error bound generated by a single gate
under the noise model $\noise$.
The algorithms used to compute these components are explained in Sections~\ref{sec:tensor} and \ref{sec:constrainedSDP},
while the soundness proof of our inference rules is given in Appendix~\ref{subsec:soundness}.

We demonstrate how these rules can be applied to %
the 2-qubit GHZ state circuit from \cref{ex:ghz}
as follows.
The program is $\widetilde{H}_\noise(q_0); \widetilde{CNOT}_\noise(q_0, q_1)$
and the input state in the density matrix form 
is $\rho=\ket{00}\bra{00}$. 
We first compute the constrained diamond norm
$\epsilon_1 = \|\widetilde{\mathcal{H}}_\noise - \mathcal{H} \|_{(\rho, 0)}$
and apply the \textsc{Gate} rule to obtain:
\[
    (\rho, 0) \vdash \widetilde{H}_\noise(q_0) \le \epsilon_1.
\]
Then, we use the tensor network approximator to compute ${TN}(\rho, H(q_0))$,
whose result is $(\hat{\rho}, \delta)$.
Using such a predicate, we compute the $(\hat\rho, \delta)$-diamond norm
$\epsilon_2 = \|\widetilde{\mathcal{CNOT}}_\noise - \mathcal{CNOT}\|_{(\hat\rho, \delta)}$.
Applying the \textsc{Gate} rule again, we obtain:
\[
    (\hat\rho, \delta)\vdash \widetilde{CNOT}_\noise(q_0, q_1) \le \epsilon_2.
\]
Finally, we apply the \textsc{Seq} rule:
\[
    (\rho, 0) \vdash \Big(\widetilde{H}_\noise(q_0); \widetilde{CNOT}_\noise(q_0, q_1) \Big) \le \epsilon_1 + \epsilon_2,
\]
which gives the error bound of the noisy program, $\epsilon_1 + \epsilon_2$.

\section{Quantum State Approximation}
\label{sec:tensor}

\framework{} %
uses tensor networks to adaptively compute
the \fix{constraints} of the input quantum state
using an approximate state $\hat{\rho}$
and its distance $\delta$ from the ideal state $\rho$.
We provide the background on tensor networks in Section~\ref{sec:tensorsub1},
present how we use tensor networks to approximate quantum states
in Section~\ref{sec:tensorsub2},
and give  examples in Section~\ref{sec:tensorsub3}.

\subsection{Tensor network}
\label{sec:tensorsub1}

\para{Tensors}
Tensors describe the multilinear relationship between sets of objects in vector spaces,
and can be represented by multi-dimensional arrays.
The \emph{rank} of a tensor indicates the dimensionality of its array representation:
vectors have rank $1$, matrices rank $2$, and
superoperators  rank $4$
(since they operate over rank $2$ density matrices).

\para{Contraction}
Tensor contraction generalizes vector inner products and matrix multiplication.
A contraction between two tensors specifies
an index for each tensor, sums over these indices, %
and produces a new tensor.
The two indices used to perform contraction must have the same range to be contracted together.
The contraction of two tensors with ranks $a$ and $b$
will have rank $a+b-2$;
for example, if we contract the first index  in  $3$-tensor $A$
and the second index  in  $2$-tensor $B$,
the output will be a 3-tensor:
$$(A \times_{1,2} B)[jkl]= \sum_{t} A[tjk] B[lt].$$

\para{Tensor product}
The tensor product is calculated like an outer product;
if two tensors have ranks $a$ and $b$ respectively, 
their tensor product 
is a rank $a+b$ tensor.
For example, the tensor product of 2-tensor $A$ and 2-tensor $B$ is a 4-tensor: 
$$(A \otimes B)[ijkl] = A[ij] B[kl].$$

\para{Tensor networks}
Tensor network (TN) representation is a graphical calculus for reasoning about tensors,
with an intuitive representation of various quantum objects.
Introduced in the 1970s by Penrose~\cite{penrose1971applications},
this notation is used in
quantum information theory~\cite{Vidal2008,Vidal2003,Evenbly2009,Verstraete2008,wood2011tensor,Schollw_ck_2005},
as well as in other  fields such as machine learning~\cite{Efthymiou2019TensorNetworkFM,NIPS2016_6211}.

\begin{figure}[t]
    \begin{subfigure}[b]{.15\textwidth}
    \centering
    \captionsetup{justification=centering}
    \begin{tikzpicture}[baseline=-0.55]
     \draw (-0.25,0) -- (0.455,0);
      \node[draw, fill=blue!30,isosceles triangle,isosceles triangle stretches,minimum height=0.05cm,  minimum width =0.05cm,inner sep=0.12pt] at (0.455,0)  {$\psi$};
    \end{tikzpicture}
    \caption{Vector \\ (rank 1)}
    \end{subfigure}
    \begin{subfigure}[b]{.15\textwidth}
    \centering
        \captionsetup{justification=centering}
     \begin{tikzpicture}[baseline=-0.55]
      \draw (-0.5,  0) -- ( 0.5, 0);
      \node [fill=blue!30] at (0, 0) [rectangle, draw]  (k) {$X$};
         \end{tikzpicture}
    \caption{Matrix \\ (rank 2)}
    \end{subfigure}
   \begin{subfigure}[b]{.15\textwidth}
    \centering
        \captionsetup{justification=centering}
     \begin{tikzpicture}[baseline=-0.55]
        \node [draw, fill=blue!30]  at(0,0) {$\mathcal{E}$};
        \begin{pgfonlayer}{bg}
           \draw (-0.5, 0.15) -- (0.5, 0.15);
           \draw (-0.5, -0.12) -- (0.5, -0.12);
        \draw (-0.5, -0.12) arc (90:270:0.1);
        \draw (0.5, -0.32) arc (-90:90:0.1);
        \end{pgfonlayer}
    \end{tikzpicture} 
    \caption{Superoperator \\ (rank 4)}
    \end{subfigure}
    \caption{{
    Tensor network representation of various tensors.
    }}
    \label{fig:simpletn}
\end{figure}
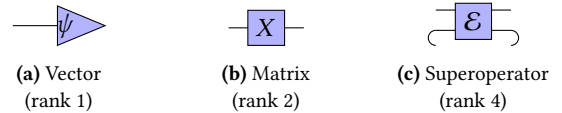

As depicted in Figure~\ref{fig:simpletn}, tensor networks consist of nodes and
edges\footnote{
Note that the shape of the nodes does not have any mathematical meaning;
they are merely used %
to distinguish different types of tensors.
}.
Each node represents a tensor,
and each edge out of the node represents an index of the tensor. 
As illustrated in Figure~\ref{fig:contraction}, %
the resulting network will itself constitute a whole tensor,
with each open-ended edge representing one index for the final tensor.
The graphical representation of a quantum program can be directly interpreted as a tensor network.
For example, the 2-qubit GHZ state circuit in Figure~\ref{fig:ghz}
can be represented by a tensor network in Figure~\ref{fig:ghz_tn}.

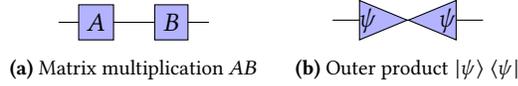
\begin{figure}[t]
    \begin{subfigure}[b]{.2\textwidth}
        \centering
        \begin{tikzpicture}[baseline=-0.55]
      \draw (-1.5,  0) -- ( 0.5, 0);
    \node [fill=blue!30] at (-1, 0) [rectangle, draw]  (a) {$A$};
     \node [fill=blue!30] at (0, 0) [rectangle, draw]  (b) {$B$};
         \end{tikzpicture}
         \caption{Matrix multiplication $AB$}
    \end{subfigure}
    \begin{subfigure}[b]{.2\textwidth}
        \centering
        \begin{tikzpicture}[baseline=-0.55]
      \draw (-1.,  0) -- ( 1., 0);
    
      \node[draw, fill=blue!30,isosceles triangle,isosceles triangle stretches,minimum height=0.05cm,  minimum width =0.05cm,inner sep=0.12pt] at (-0.53, 0)  {$\psi$};
      \node[draw, fill=blue!30,isosceles triangle,isosceles triangle stretches,minimum height=0.05cm,  minimum width =0.05cm,inner sep=0.12pt,shape border rotate=180] at (0.53, 0)  {$\psi$};
         \end{tikzpicture}
         \caption{Outer product $\ket{\psi}\bra{\psi}$}
    \end{subfigure}
    \caption{Tensor network representation for two matrix operations.
    In general, tensor contractions are represented by linking edges, and tensor products by juxtaposition.
    }
    \label{fig:contraction}
\end{figure}

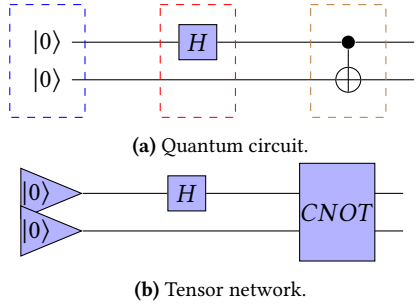
\begin{figure}[t]
    \begin{subfigure}[b]{.4\textwidth}
    \centering
    \begin{tikzpicture}
    \tikzstyle{operator} = [draw,fill=blue!30] 
    \tikzstyle{phase} = [fill,shape=circle,minimum size=5pt,inner sep=0pt]
    \tikzstyle{surround} = [fill=blue!10,thick,draw=black,rounded corners=2mm]
    \tikzset{XOR/.style={draw,circle,append after command={
        [shorten >=\pgflinewidth, shorten <=\pgflinewidth,]
        (\tikzlastnode.north) edge (\tikzlastnode.south)
        (\tikzlastnode.east) edge (\tikzlastnode.west)
        }
    }
}
    \node at (-1,0) (q1) {$\ket{0}$};
    \node at (-1,-0.5) (q2) {$\ket{0}$};

    \node[operator] (op11) at (1,0) {$H$} edge [-] (q1);

    \node[phase] (phase11) at (3,0) {} edge [-] (op11);
    \node[XOR] (phase12) at (3,-0.5) {} edge [-] (q2);
    \draw[-] (phase11) -- (phase12);

    \node (end1) at (4,0) {} edge [-] (phase11);
    \node (end2) at (4,-0.5) {} edge [-] (phase12);
    
    \draw [dashed, blue] (-1.5, 0.5) -- (-0.5, 0.5) -- (-0.5, -1) -- (-1.5,-1) -- (-1.5, 0.5);
    \draw [dashed, red] (0.5, 0.5) -- (1.5, 0.5) -- (1.5, -1) -- (0.5,-1) -- (0.5, 0.5);
    \draw [dashed, brown] (2.5, 0.5) -- (3.5, 0.5) -- (3.5, -1) -- (2.5,-1) -- (2.5, 0.5);
    \end{tikzpicture}
    \caption{Quantum circuit.}
    \end{subfigure}
    \begin{subfigure}[b]{.4\textwidth}
        \centering
    \begin{tikzpicture}
    \tikzstyle{operator} = [draw,fill=blue!30] 
    \tikzstyle{phase} = [fill,shape=circle,minimum size=5pt,inner sep=0pt]
    \tikzstyle{surround} = [fill=blue!10,thick,draw=black,rounded corners=2mm]
    \tikzset{XOR/.style={draw,circle,append after command={
        [shorten >=\pgflinewidth, shorten <=\pgflinewidth,]
        (\tikzlastnode.north) edge (\tikzlastnode.south)
        (\tikzlastnode.east) edge (\tikzlastnode.west)
        }
    }
}
    \node[draw, fill=blue!30,isosceles triangle,isosceles triangle stretches,minimum height=0.05cm,  minimum width =0.05cm,inner sep=0.12pt] at (-1,0) (q1) {$\ket{0}$};
    \node[draw, fill=blue!30,isosceles triangle,isosceles triangle stretches,minimum height=0.05cm,  minimum width =0.05cm,inner sep=0.12pt] at (-1,-0.5) (q2) {$\ket{0}$};

    \node[operator] (op11) at (1,0) {$H$} edge [-] (q1);

    \node[phase] (phase11) at (3,0) {} edge [-] (op11);
    \node[XOR] (phase12) at (3,-0.5) {} edge [-] (q2);
    \draw[-] (phase11) -- (phase12);

    \node (end1) at (4,0) {} edge [-] (phase11);
    \node (end2) at (4,-0.5) {} edge [-] (phase12);
    
    \draw [fill=blue!30]  (2.5, 0.4) rectangle  (3.5, -0.9) node[pos=.5] (d) {$CNOT$};
    
    \end{tikzpicture}
    \caption{Tensor network.}
    \end{subfigure}
    \caption{{
        The GHZ state, represented as a quantum circuit (a) and a tensor network (b).
        When we evaluate the output of the circuit, we can see that
        the input state $\ket{00}$ (enclosed in the dashed blue box),
        the $H$ gate $H\otimes I$ (enclosed in the dashed red box),
        and the $CNOT$ gate (enclosed in the dashed brown box). When evaluating the tensor network in (b), the 
         output  is the same as the program output,
        $\big(\ket{00} + \ket{11}\big)/\sqrt{2}$.
    }}
    \label{fig:ghz_tn}
\end{figure}

\tabulinesep=0.9mm
\begin{table*}[t]
\centering
\setlength{\tabcolsep}{2pt}
\addtolength{\leftskip} {-3cm} %
    \addtolength{\rightskip}{-3cm}
\begin{tabu}{c|c|c|c}
 & \textsc{Gate} & \textsc{Superoperator} & \textsc{Singular Value}  \\
  & \textsc{Contraction} & \textsc{Application} & \textsc{Decomposition} \\
  \hline
Tensor network &
\begin{tikzpicture}[baseline=-4]
    \draw (-0.5,0) -- (-0.25,0);
    \node[draw, fill=blue!30] at (0,0) {$U$};
    \draw (0.25,0) -- (0.6,0);
   \node[draw, fill=blue!30,isosceles triangle,isosceles triangle stretches,minimum height=0.05cm,  minimum width =0.05cm,inner sep=0.12pt] at (0.7,0)   {$\psi$};
   \draw (1.2,0) -- (1.4, 0);
 \node at (1.73,-0.01) {\large $\rightarrow$};
     \draw (2.05,0) -- (2.305,0);
          \node[draw, fill=blue!30,isosceles triangle,isosceles triangle stretches,minimum height=0.05cm,  minimum width =0.05cm,inner sep=0.12pt]at (2.305,0)  {$\phi$};
    \draw (2.8,0) --(3.0,0);
    \end{tikzpicture}
    &
    \begin{tikzpicture}[baseline=-10.55]
     \node [draw, fill=blue!30]  at(0,0) {$\mathcal{E}$};
      
    \node [draw, fill=blue!30] at (0.,-0.55)  {$\rho$};
    \node at (0.8,-0.2) {\large $\rightarrow$};
    \node [draw, fill=blue!30] at (1.555,-0.2)  {$\hat{\rho}$};
         \begin{pgfonlayer}{bg}
          \draw (-0.5, 0.15) -- (0.5, 0.15);
           \draw (-0.3, -0.12) -- (0.3, -0.12);
        \draw (-0.3, -0.12) arc (90:270:0.2);
        \draw (0.3, -0.52) arc (-90:90:0.2);
         \draw(-0.3, -0.52) -- (0.3, -0.52);
         \draw(1.15, -0.2) -- (1.95, -0.2);
         \end{pgfonlayer}
    \end{tikzpicture}
    &
        \begin{tikzpicture}[baseline=-4]
      \draw (-0.4, 0) -- ( 0.4, 0);
      \node[draw, fill=blue!30] at (0, 0) (k) {$M$};
     \node at (.8,-0.01) {\large $\rightarrow$}; 
      \draw (1.2, 0) -- (3.2, 0);

     \node[draw, fill=blue!30] at (1.6, 0) (k) {$U$};
     \node [fill=red!30, draw, aspect=.6, inner sep=0.5pt, diamond] at (2.2,0)  (m) {$\sigma$};
     \node[draw, fill=blue!30, text=blue!30] at (2.8, 0) (k) {$V$};     \node at (2.8, 0) {$V^\dagger$};

         \end{tikzpicture}
    \\
Dirac notation &  $U\ket{\psi} = \ket{\phi}$ & $\mathcal{E}(\rho) = \hat{\rho} $ &  $M=\sum_j \sigma_j U\ket{j}\bra{j} V^\dagger$ 
\\
Rank  & 1 & 2 & 2
\end{tabu}
\vspace*{0.2cm}
\caption{Examples of tensor network transformations for basic quantum operations.}
\label{tab:tensorops}
\end{table*}

\para{Transforming tensor networks}
To speed up the evaluation of a large tensor network,
we can apply reduction rules
to transform and simplify  the network structure.
In \cref{tab:tensorops}, we summarize some common reduction rules we use.
The \textsc{Gate Contraction} rule transforms a vector  $\psi$ and a matrix $U$ connected to it into a new vector $\phi$ that is the product of $U$ and $\psi$. The \textsc{Superoperator Application} rule transforms a superoperator 
$\mathcal{E}$ and a matrix $\rho$ connected to it into a matrix
$\hat{\rho}$ that represents the application of the superoperator $\mathcal{E}$ to $\rho$. The \textsc{Singular Value Decomposition (SVD)} rule transforms a matrix $M$ into the product of three matrices: $U$, $\Sigma$, and $V^\dag$, \fix{where $\Sigma$ is a diagonal matrix whose diagonal entries are the singular values $\sigma_1, \ldots, \sigma_n$ of $M$. This special matrix $\Sigma = \sum_{j}\sigma_j \ket{j}\bra{j}$ is graphically represented by a diamond.} 
By dropping small singular values in the diagonal matrix $\Sigma$,
we can obtain a simpler tensor network which closely approximates the original one.

\subsection{Approximate quantum states}
\label{sec:tensorsub2}
In this section, we describe our  tensor network approximator algorithm computing ${TN}(\rho, P) = (\hat\rho, \delta)$, such that the trace distance between our approximation $\hat{\rho}$ and the perfect output $[\![P]\!](\rho)$ satisfies ${T}\big(\hat\rho, [\![P]\!](\rho)\big) \le \delta$.
At each stage of the algorithm, we use Matrix Product State (MPS) tensor networks~\cite{mps}
to approximate quantum states.
This class of tensor networks uses $2n$ matrices to represent $2^n$-length vectors, greatly reducing the computational cost.
MPS tensor networks take a size $w$ as an argument,
which determines the space of representable states.
When $w$ is not large enough to represent all possible quantum states,
the MPS  is an \emph{approximate}  quantum state
whose approximation bound depends on $w$.
The  MPS representation with size $w$ of a quantum state $\psi$ (represented as a vector)  is:
\begin{align*}
\ket{\psi}_{\text{MPS}} :=
\sum_{i_1, \ldots, i_n} A_1^{(i_1)} A_2^{(i_2)} \cdots A_n^{(i_n)}\ket{i_1 i_2 \cdots i_n},
\end{align*}
\noindent{}where $A_1^{(i_1)}$ is a row vector of dimension $w$,
$A_2^{(i_2)}, \ldots, A_{n-1}^{(i_{n-1})}$ are $w \times w$ matrices,
and $A_n^{(i_n)}$ is a column vector of dimension $w$.
We use $i_j$ to represent the value of a  basis
$\ket{i_1 i_2 \cdots i_n}$ at position $j$, which 
can be $0$ or $1$.
For example, to represent the 3-qubit state $\big(\ket{000} + \ket{010} + \ket{001}\big)/3$ in MPS, 
we must find matrices $A_1^{(0)}, A_1^{(1)}, A_2^{(0)}, A_2^{(1)}, A_3^{(0)}, A_3^{(1)}$ such that
$$A_1^{(0)} A_2^{(0)} A_3^{(0)} = A_1^{(0)} A_2^{(1)} A_3^{(0)} =  A_1^{(0)} A_2^{(0)} A_3^{(1)} = \frac{1}{3},$$
while $ A_1^{(i_1)} A_2^{(i_2)} A_3^{(i_3)} = 0$ for all $(i_1, i_2, i_3) \ne (0, 0, 0)$, $(0, 1, 0)$, or $(0, 0, 1)$.

$A_i^{(0)}$ and $A_i^{(1)}$ can be taken together as a 3-tensor $A_i$ (\fix{$A_1$} and $A_n$ are 2-tensors)
where the superscript is taken as the third index besides the two indices of the matrix.
Overall, the MPS representation can be seen as a tensor network, as shown in Figure~\ref{fig:mps}. $A_1, \ldots, A_n$ are linked together in a line, while
$i_1, \ldots, i_n$ are open wires. %

\begin{figure}[t]
    \centering
     \vspace{-5pt}
   \begin{tikzpicture}
     \draw [fill=blue!30] (  0,0) circle [radius=0.24] node (a) {$A_1$};
     \draw [fill=blue!30] (1.1,0) circle [radius=0.24] node (b) {$A_2$};
     \draw [fill=blue!30] (2.2,0) circle [radius=0.24] node (c) {$A_3$};
     \draw [fill=blue!30] (3.3,0) circle [radius=0.24] node (e) {$A_4$};
     \draw [fill=blue!30] (4.4,0) circle [radius=0.24] node (f) {$A_5$};
     \draw [fill=blue!30] (5.5,0) circle [radius=0.24] node (g) {$A_6$};
    
    \node at (0, -0.7) {$i_1$};
     \node at (1.1, -0.7) {$i_2$};
      \node at (2.2, -0.7) {$i_3$};
       \node at (3.3, -0.7) {$i_4$};
        \node at (4.4, -0.7) {$i_5$};
         \node at (5.5, -0.7) {$i_6$};
   \begin{pgfonlayer}{bg}    %
    \draw (0,0) -- (5.5,0);
    \draw (b) -- (1.1, -0.5);
    \draw (a) -- (0, -0.5);
    \draw (c) -- (2.2, -0.5);
    \draw (f) -- (4.4, -0.5);
    \draw (e) -- (3.3, -0.5);
    \draw (g) -- (5.5, -0.5);
    \end{pgfonlayer}
   \end{tikzpicture}
    
    \caption{The MPS representation of six qubits.}
    \label{fig:mps}
\end{figure}
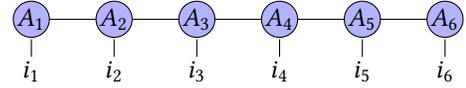

Our approximation algorithm starts by initializing the MPS to the input state in vector form.
Then, for each gate of the quantum program, we apply it to the MPS to get the intermediate state at this step
and compute the distance between MPS and the ideal state.
Since MPS only needs to maintain $2n$ tensors, i.e.,
$A_1^{(0)}$, $A_1^{(1)}$, $A_2^{(0)}$, $A_2^{(1)}$, $\cdots$,
$A_n^{(0)}$, $A_n^{(1)}$,
this procedure can be performed efficiently with
polynomial \fix{running time}.
After applying all quantum gates,
we obtain an MPS that approximates the output state of the quantum program,
as well as an approximation bound by summing together all accumulated approximation errors
incurred by the approximation process.
Our approximation algorithm consists of the following stages.

\para{Initialization}
Let $\ket{s_1 s_2 \cdots s_n}$ be the input state for an $n$-qubit quantum circuit.
For all $k \in [1, n]$, we initialize $A_k^{(s_k)} = E$ and $A_k^{(1-s_k)} = 0$, where $E$ is the matrix that 
$E_{1,1} = 1$ and $E_{i,j} = 0$ \fix{for all $i \ne 1$ or $ j \ne 1$}.

\para{Applying 1-qubit gates}
Applying a 1-qubit gate on an MPS always results in an MPS
and thus does not incur any approximation error.
For a single-qubit gate $G$ on qubit $i$, we update the tensor $A_i$
to $A'_i$ as follows:
\[
    {A'}_i^{(s)} = \sum_{s'\in \{0, 1\}} G_{s s'} A_i^{(s')}
 \qquad \text{for } s = 0 \text{ or } 1.
 \]
\noindent{}In the tensor network representation, such application amounts to contracting the tensor for the gate
with $A_i$ (see Figure~\ref{fig:one_qubit}).

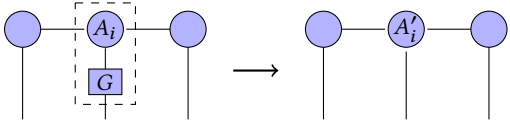
\begin{figure}[t]
    \centering
   \begin{tikzpicture}
     \draw [fill=blue!30] (0,0) circle [radius=0.24] node (a) {};
     \draw [fill=blue!30] (1.1,0) circle [radius=0.24] node (b) {\small{$A_i$}};
     \draw [fill=blue!30] (2.2,0) circle [radius=0.24] node (c) {};
     \draw [fill=blue!30] (1.1-0.22,-0.7-0.17) rectangle  (1.1+0.22,-0.7+0.17) node[pos=.5] (d) {\small{$G$}};
    
     \draw [fill=blue!30] (4,0) circle [radius=0.24] node (e) {};
     \draw [fill=blue!30] (5.1,0) circle [radius=0.24] node (f) {\small{$A'_i$}};
     \draw [fill=blue!30] (6.2,0) circle [radius=0.24] node (g) {};
    
    \draw [->, thick] (2.8, -0.55) -- (3.4, -0.55);
    
    \draw [dashed] (0.7, 0.35) -- (1.5, .35) -- (1.5, -1.) --(.7, -1.) --cycle;
   \begin{pgfonlayer}{bg}    %
    \draw (a.east) -- (b.west);
    \draw (b) -- (c);
    \draw (b) -- (1.1, -1.2);
    \draw (a) -- (0, -1.2);
    \draw (c) -- (2.2, -1.2);
    
    \draw (e.east) -- (f.west);
    \draw (f) -- (g);
    \draw (f) -- (5.1, -1.2);
    \draw (e) -- (4, -1.2);
    \draw (g) -- (6.2, -1.2);
    \end{pgfonlayer}
   \end{tikzpicture}

    \caption{Applying a 1-qubit gate to an MPS.
    We contract the MPS node for the qubit and the gate (in the dashed box),
    resulting in another 3-tensor MPS node.}
    \label{fig:one_qubit}
\end{figure}

\begin{figure*}[t]
    \centering
   \begin{tikzpicture}
     \draw [fill=blue!30] (0,0) circle [radius=0.25] node (a) {};
     \draw [fill=blue!30] (0.9,0) circle [radius=0.25] node (b) {\scriptsize{$A_i$}};
     \draw [fill=blue!30] (1.8,0) circle [radius=0.25] node (c) {\scriptsize{$A_{i+1}$}};
     \draw [fill=blue!30] (2.7,0) circle [radius=0.25] node (h) {};
     \draw [fill=blue!30] (0.9-0.2,-0.7-0.20) rectangle  (1.8+0.2,-0.7+0.15) node[pos=.5] (d) {\scriptsize$G$};
    
     \draw [fill=blue!30] (4,0) circle [radius=0.25] node (e) {};
     \draw [fill=blue!30] (5.1-0.25, -0.20) rectangle (5.1+.25, 0.20)  node[pos=0.5] (f) {\scriptsize$M'$};
     \draw [fill=blue!30] (6.2,0) circle [radius=0.25] node (g) {};

     \draw [fill=blue!30] (9.3-1.35,0) circle [radius=0.25] node (i) {};
     \draw [fill=blue!30] (9.3-0.65,0) circle [radius=0.25] node (j) {\scriptsize$U$};
     \draw [fill=blue!30] (9.3+.65,0) circle [radius=0.25] node (k) {\scriptsize$V^\dag$};
     \draw [fill=blue!30] (9.3+1.35,0) circle [radius=0.25] node (l) {};
     \node [fill=red!30, draw, aspect=.6, inner sep=2pt, diamond] at (9.3,0)  (m) {};

     \draw [fill=blue!30] (12.3+0,0) circle [radius=0.25] node (n) {};
     \draw [fill=blue!30] (12.3+0.9,0) circle [radius=0.25] node (o) {\scriptsize{$A'_i$}};
     \draw [fill=blue!30] (12.3+1.8,0) circle [radius=0.25] node (p) {\scriptsize{$A'_{i+1}$}};
     \draw [fill=blue!30] (12.3+2.7,0) circle [radius=0.25] node (q) {};

    \draw [->, thick] (3.0, -0.55) -- (3.6, -0.55) node[midway, above] {(i)};
    \draw [->, thick] (6.8, -0.55) -- (7.4, -0.55) node[midway, above] {(ii)};
      \draw [->, thick] (11.2, -0.55) -- (11.8, -0.55) node[midway, above] {(iii)};
    \draw [dashed] (0.55, 0.35) -- (2.15, .35) -- (2.15, -1.02) --(.55, -1.02) --cycle;
    \draw [dashed] (5.1+0.35, 0.35) -- (5.1-0.35, .35)  --(5.1-0.35, -0.35) --(5.1+0.35, -0.35) --cycle;
    \draw [dashed] (9.3+1.0, 0-0.4) -- (9.3-1.0, 0-0.4)  --(9.3-1.0, 0+.4) --(9.3+1.0, 0+.4) --cycle;
    
   \begin{pgfonlayer}{bg}    %
    \draw (a) -- (h);
    \draw (b) -- (.9, -1.2);
    \draw (a) -- (0, -1.2);
    \draw (c) -- (1.8, -1.2);
     \draw (h) -- (2.7, -1.2);
    
    \draw (e.east) --  (g);
    \draw (5.1-0.09, 0) -- (5.1-0.09, -1.2);
    \draw (5.1+0.09, 0) -- (5.1+0.09, -1.2);
    \draw (e) -- (4, -1.2);
    \draw (g) -- (6.2, -1.2);
    
    \draw (i) -- (j);
    \draw (k) -- (l);
    \draw (9.3-0.65, -0.04) -- (9.3+0.65, -0.04);
    \draw (9.3-0.65, 0.04) -- (9.3+0.65, 0.04);
    \draw (i) -- (9.3-1.35, -1.2);
    \draw (j) -- (9.3-0.65, -1.2);
    \draw (k) -- (9.3+.65, -1.2);
     \draw (l) -- (9.3+1.35, -1.2);

    \draw (n) -- (q);
    \draw (n) -- (12.3+0, -1.2);
    \draw (o) -- (12.3+0.9, -1.2);
    \draw (p) -- (12.3+1.8, -1.2);
     \draw (q) -- (12.3+1.35+1.35, -1.2);
    
    \end{pgfonlayer}
   \end{tikzpicture}
    \caption{Applying a 2-qubit gate on two adjacent qubits to the MPS,
    \change{via (i) node contraction, (ii) singular value decomposition, and (iii) singular value truncation with re-normalization.}
    }
    \label{fig:two_qubit}
\end{figure*}
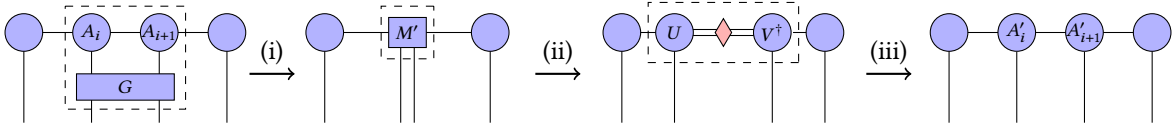

\para{Applying 2-qubit gates}
If we are applying a 2-qubit gate $G$ on two adjacent qubits $i$ and $i + 1$,
we only need to modify $A_i$ and $A_{i + 1}$.
We first contract $A_i$ and $A_{i + 1}$ to get \fix{a} $2w \times 2w$ matrix $M$:
\[
    \begin{bmatrix} A_i^{(0)} \\ A_i^{(1)} \end{bmatrix}
    \begin{bmatrix} A_{i+1}^{(0)} & A_{i+1}^{(1)} \end{bmatrix}
    =
    \begin{bmatrix}M_{00} & M_{01} \\ M_{10} & M_{11} \end{bmatrix}
    =
    M.
\]
\noindent{}Then, we apply the 2-qubit gate to it:
\[
    M_{ij}' = \sum_{k, l} G_{ijkl} M_{kl}.
\]
\noindent{}We then need to decompose this new matrix
$M'$ back into two tensors.
We first apply the 
\textsc{Singular Value Decomposition} rule on the contracted matrix:
\[
    M' = U \Sigma V^\dag.
\]
\noindent{}When $w$ is not big enough to represent all possible quantum states, $M'$ introduces approximation errors
and may not be a contraction of two tensors.
Thus, we \emph{truncate} the lower half of the singular values in $\Sigma$, enabling the tensor decomposition while reducing
the error: %
\[
    \Sigma \approx \begin{bmatrix} \Sigma' & 0 \\ 0 & 0\end{bmatrix}.
\]
\noindent{}Therefore, we arrive at a new MPS  whose
new tensors $A'_i$ and $A'_{i+1}$ are calculated as follows:
\[
    \begin{bmatrix} A_i^{(0)'} & *\\ A_i^{(1)'} & *  \end{bmatrix}  = U,
    \begin{bmatrix} A_{i+1}^{(0)'} & A_{i+1}^{(1)'} \\ * & * \end{bmatrix} = \Sigma' V,
\]
where $*$ denotes \fix{the part that we truncate}.
After truncation, we renormalize the state to a norm-$1$ vector.

Figure~\ref{fig:two_qubit} shows the  above procedure 
in tensor network form
by (1) first applying \textsc{Gate Contracting} rule for
$A_i$, $A_{i + 1}$ and $G$,
(2) using \textsc{Singular Value Decomposition} rule to decompose
the contracted tensor,
 (3) truncating the internal edge to width $w$,
 and finally (4) calculating the updated $A'_i$ and $A'_{i + 1}$. 
If we want to apply a 2-qubit gate to non-adjacent qubits,
we add swap gates to move the two qubits together before
applying the gate to the now adjacent pair of qubits.

\para{Bounding approximation errors}
When applying 2-qubit gates, we compute an MPS to approximate the gate application.
Each time we do so, we must bound the error due to this approximation.
Since the truncated values themselves comprise an MPS state,
we may determine the error by simply calculating the trace distance between
the states before and after truncation. %

The trace distance of two MPS states can be calculated from the inner product
of these two MPS:
$$\delta := \big|\ket{\phi}\bra{\phi} - \ket{\psi}\bra{\psi}\big\|_1 = 2 \sqrt{1 - |\braket{\phi|\psi}|^2}. $$
The inner product of two  states
$\ket{\psi}$ and $\ket{\phi}$ (represented using 
$A$ and $B$ in their MPS forms) is defined as follows:
\begin{align*}
    \braket{\psi|\phi} %
    =& \sum_{i_1, \ldots, i_n} \Big\langle A_1^{(i_1)} \cdots A_n^{(i_n)}, 
                                    B_1^{(i_1)} \cdots B_n^{(i_n)}
        \Big\rangle.
\end{align*}
\noindent{}Figure~\ref{fig:mps_inner_product} shows its tensor network graphical representation.

\begin{figure}[t]
    \centering
   \begin{tikzpicture}
     \draw [fill=blue!30] (0,0) circle [radius=0.24] node (a) {$A_1$};
     \draw [fill=blue!30] (1.1,0) circle [radius=0.24] node (b) {$A_2$};
     \draw [fill=blue!30] (2.2,0) circle [radius=0.24] node (c) {$A_3$};
     \draw [fill=blue!30] (3.3,0) circle [radius=0.24] node (e) {$A_4$};
     \draw [fill=blue!30] (4.4,0) circle [radius=0.24] node (f) {$A_5$};
     \draw [fill=blue!30] (5.5,0) circle [radius=0.24] node (g) {$A_6$};
     
     \draw [fill=blue!30] (0,-1) circle [radius=0.24] node (a2) {$B_1$};
     \draw [fill=blue!30] (1.1,-1) circle [radius=0.24] node (b2) {$B_2$};
     \draw [fill=blue!30] (2.2,-1) circle [radius=0.24] node (c2) {$B_3$};
     \draw [fill=blue!30] (3.3,-1) circle [radius=0.24] node (e2) {$B_4$};
     \draw [fill=blue!30] (4.4,-1) circle [radius=0.24] node (f2) {$B_5$};
     \draw [fill=blue!30] (5.5,-1) circle [radius=0.24] node (g2) {$B_6$};

   \begin{pgfonlayer}{bg}    %
    \draw (0,0) -- (5.5,0);
    \draw (0,-1) -- (5.5,-1);
    \draw (b) -- (1.1, -1.2);
    \draw (a) -- (0, -1.2);
    \draw (c) -- (2.2, -1.2);
    \draw (f) -- (4.4, -1.2);
    \draw (e) -- (3.3, -1.2);
    \draw (g) -- (5.5, -1.2);
    \end{pgfonlayer}
   \end{tikzpicture}
    \caption{
        Tensor network representation of the inner product of two MPSs.
        An open wire of one MPS is linked with an open wire of another,
        which denotes the summation over $i_1, \ldots, i_n$.
    }
    \label{fig:mps_inner_product}
\end{figure}
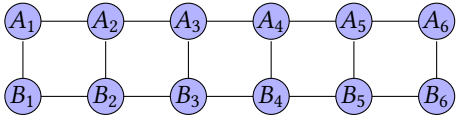

In our approximation algorithm,
we can iteratively calculate \fix{from qubit $1$ to qubit $n$ the distance} by first determining:
\[
    D_1 = A_1^{(0)} B_1^{(0)\dag} + A_1^{(1)} B_1^{(1)\dag}.
\]
Then, we repeatedly apply tensors to the rest of qubits:%
\[
    D_i = A_i^{(0)} D_{i-1} B_i^{(0)\dag} + A_i^{(1)} D_{i-1} B_i^{(1)\dag},
\]
leading us to the final result of $D_n = \braket{\psi|\phi}$.
In the  tensor network graphical representation,
this algorithm is a left-to-right contraction, as shown in Figure~\ref{fig:mps_contract}.

\begin{figure*}[t]
    \centering
   \begin{tikzpicture}
     \draw [fill=blue!30] (-1,0) circle [radius=0.24] node (a) {$A_1$};
     \draw [fill=blue!30] (0.1,0) circle [radius=0.24] node (b) {$A_2$};
     \draw [fill=blue!30] (1.2,0) circle [radius=0.24] node (c) {$A_3$};
     \draw [fill=blue!30] (4,0) circle [radius=0.24] node (e) {$D_1$};
     \draw [fill=blue!30] (5.1,0) circle [radius=0.24] node (f) {$A_2$};
     \draw [fill=blue!30] (6.2,0) circle [radius=0.24] node (g) {$A_3$};

     \draw [fill=blue!30] (9-0.5,0) circle [radius=0.24] node (xf) {$D_2$};
     \draw [fill=blue!30] (9+1.1-0.5,0) circle [radius=0.24] node (xg) {$A_3$};
     
    \draw [fill=blue!30] (9+1.1+1.8,0) circle [radius=0.24] node (xc) {$D_3$};
     
     \node at(3.2,-0.5) (x1) {\large $\rightarrow$};
     \node at(7.8,-0.5) (x3) {\large $\rightarrow$};
     \node at(14,-0.6) (x3) {\large $\xrightarrow[n-3\text{ steps}]{}$};
     \node at(11,-0.5) (x2) {\large $\rightarrow$};
     
     \node at(2,0) {$\cdots$};
     \node at(2,-1) {$\cdots$};
     \node at(7,0) {$\cdots$};
     \node at(7,-1) {$\cdots$};
     \node at(6.65+4.25-0.5,-0) {$\cdots$};
     \node at(6.65+4.25-0.5,-1) {$\cdots$};
     \node at(2.5+10.2,0) {$\cdots$};
     \node at(2.5+10.2,-1) {$\cdots$};
     
     \draw [fill=blue!30] (-1,-1) circle [radius=0.24] node (a2) {$B_1$};
     \draw [fill=blue!30] (0.1,-1) circle [radius=0.24] node (b2) {$B_2$};
     \draw [fill=blue!30] (1.2,-1) circle [radius=0.24] node (c2) {$B_3$};

     \draw [fill=blue!30] (5.1,-1) circle [radius=0.24] node (f2) {$B_2$};
     \draw [fill=blue!30] (6.2,-1) circle [radius=0.24] node (g2) {$B_3$};
     
     \draw [fill=blue!30] (9+1.1-0.5,-1) circle [radius=0.24] node (xg2) {$B_3$};
     
     \draw [fill=blue!30] (15.5,-0) circle [radius=0.24] node (xgn) {$D_n$};
     
   \begin{pgfonlayer}{bg}    %
    \draw (-1,0) -- (1.7,0);
    \draw (-1,-1) -- (1.7,-1);
    \draw (b) -- (0.1, -1.2);
    \draw (a) -- (-1, -1.2);
    \draw (c) -- (1.2, -1.2);
    \draw (f) -- (f2.center);
    \draw (e.center) -- (f2.center);
    \draw (e.center) -- (6.7, 0);
    \draw (f2.center) -- (6.7,-1);
    \draw (g.center) -- (g2.center);
    \draw (xf.center) -- (6.35+4.25-0.5,0);
    \draw (xf.center) -- (xg2.center);
    \draw (xg.center) -- (xg2.center);
    \draw (xg2.center) -- (6.35+4.25-0.5,-1);
    \draw (xc.center) -- (2.2+10.2, 0);
    \draw (xc.center) -- (2.2+10.2,-1);
    \end{pgfonlayer}
   \end{tikzpicture}
    \caption{Contraction of the inner product of two MPS. We first contract $A_1$ and $B_1$ to get $D_1$. Then contract $D_1$, $A_2$ and $B_2$ to $D_2$. And then $D_2, A_3$ and $B_3$ to $D_3$. Repeating this process will result in a single tensor node $D_n$, i.e., the final answer. %
    }
    \label{fig:mps_contract}
\end{figure*}
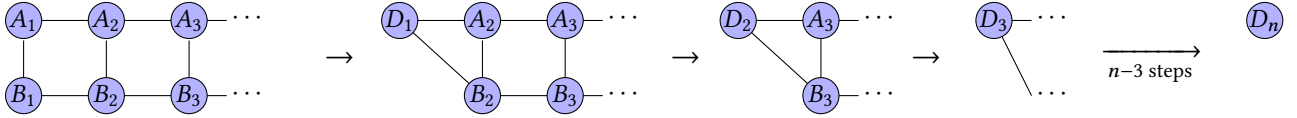

Given the calculated distance of each step,
we must combine them to obtain the overall approximation error.
For some arbitrary quantum program with $t$ 2-qubit gates,
let the truncation errors be $\delta_1, \delta_2, \ldots, \delta_t$
when applying the 2-qubit gates $g_1, g_2,...,g_t$,
the final approximation error is $\delta = \sum_{i=1}^t \delta_i$.

\fix{To see why}, we consider the approximation of one 2-qubit gate.
Let $\ket{\psi}$ denote some quantum state and \fix{$\ket{\hat{\psi}}$} its approximation
with bounded error $\delta_0$.
After applying a 2-qubit gate $G$ to the approximate MPS state,
we obtain the truncated result $\ket{{\phi}}$ with bounded error $\delta_1$.
We now have:
\begin{align}
    \|G \ket{\psi} - \ket{\phi} \| &\le \|G \ket{\psi} -
        G\ket{\hat{\psi}} \| \ +\ \|G\ket{\hat{\psi}} - \ket{\phi}\| \nonumber\\
    &= \| \ket{\psi} - \ket{\hat{\psi}} \|\ +\ \|G\ket{\hat{\psi}} - \ket{\phi}\|\nonumber \\
    &= \delta_0 + \delta_1, \label{eq:additivity}
\end{align}
where $\|\ket{\psi} - \ket{\phi}\| = \|\ket{\psi}\bra{\psi} - \ket{\phi}\bra{\phi}\|_1$.
The inequality holds because of the triangular inequality of quantum state distance
and the fact that $G$ is unitary, thus preserving the trace norm.
Repeating this for each step,
we know that the total approximation error is bounded by the sum of all approximation errors. 

\para{Supporting branches}
\fix{Due to the deferred measurement principle~\cite{NielsenChuang},
measurements can always be delayed to the end of the program.
Thus, in-program measurements are not required for quantum program error analysis. 
Our approach can also directly support \texttt{if} statements by calculating an MPS for each branch.
When we apply the measurement on the $i$-th qubit, we obtain the collapsed states
by simply setting $A_i^{(0)}$ or $A_i^{(1)}$ to the zero matrix,
obtaining MPS tensor networks corresponding to measurements of $0$ and $1$.
Using these states, we continue to evaluate the subsequent MPS in each branch separately.
We cannot merge the measured states once they have diverged, so we must duplicate any code
sequenced after the branch and compute the approximated state separately;
the number of intermediate MPS representations we must compute is thus the number of branches.
The overall approximation error is taken to be the sum of approximation errors incurred on all branches.
Note that the number of branches may be exponential to
the number of \texttt{if} statements.
}

\para{Complexity analysis}
The running time of all the operations above scales polynomially  with respect to
the MPS size $w$, number of qubits $n$, number of branches $b$, and number of gates $m$ in the program.
To be precise, applying a 1-qubit gate only requires one matrix addition 
with a $O(w^2)$ time complexity.
Applying a 2-qubit gate requires matrix multiplication and SVD
with a $O(w^3)$ time complexity.
Computing inner product of two MPS (e.g. for contraction)
requires $O(n)$ of matrix multiplications,
incurring an overall running time of $O(nw^3)$. 
Since the algorithm \fix{scans} all $m$ gates in the program,
the overall  time complexity is $O(bmnw^3)$.

Although a \emph{perfect} approximation (i.e., a full simulation)
requires an MPS size that scales exponentially with respect to the number of qubits \fix{(i.e., $2^n$ size when there are $n$ qubits)},
our approximation algorithm allows \framework{} to be configured with smaller MPS sizes,
sacrificing some precision in favor of efficiency
and enabling its practical use for real-world quantum programs.

\para{Correctness}
From the quantum program semantics defined in Figure~\ref{fig:semantics},
we know that we can compute the output state by applying all the program's gates 
in sequence.
Following Equation~\eqref{eq:additivity}, we know that the total error bound for our approximation algorithm
is bounded by the sum of each step's bound.
Thus, we can conclude that our algorithm correctly approximates
the output state and correctly bounds the approximation error in doing so.
\begin{theorem}
\label{thm:tn}
Let the output of our approximation algorithm be $(\hat\rho, \delta) = \mathrm{TN}(\rho, P)$.
The trace distance between the approximation and perfect output is bound by $\delta$:
\[
    \fix{\Big\|[\![P]\!](\hat\rho) - [\![P]\!](\rho)\Big\|_1} \le \delta.
\]
\end{theorem}

\subsection{Example: GHZ circuit}
\label{sec:tensorsub3}

We revisit the GHZ circuit in Figure~\ref{fig:ghz} to walk through
how we approximate quantum states with MPS tensor networks. 
The same technique can be applied to larger and more complex quantum circuits,
discussed in Section~\ref{sec:evaluation}.

\para{Approximation using 2-wide MPS}
Since the program only contains two qubits,
an MPS with size $w=2$ can already perfectly represent
all possible quantum states such that no approximation error
will be introduced.
Assume the input state is $\ket{00}$.
First, we initialize all the tensors based on the input state $\ket{00}$:
\[
    A_1^{(0)} = [1, 0],\ A_1^{(1)} = [0, 0],\  A_2^{(0)} = [1, 0]^T,\ 
      A_2^{(1)}  = [0, 0]^T.
\]

Then, we apply the first $H$ gate to qubit 1,
changing only $A_1^{(0)}$ and $A_1^{(1)}$:

\begin{center}
    \begin{tabu}{cc}
        $A_1^{(0)} = [1, 0] / \sqrt{2}$, &
        $A_1^{(1)} = [1, 0]/ \sqrt{2}$. 
    \end{tabu}  
\end{center}

To apply the CNOT gate on qubit $1$ and $2$,
we first compute matrix $M$ and $M'$:
\[
    M = \begin{bmatrix}1/\sqrt{2} & 0  \\ 1/\sqrt{2} & 0\end{bmatrix},\quad
    M' = \begin{bmatrix}1/\sqrt{2} & 0  \\ 0 &  1/\sqrt{2}\end{bmatrix}.
\]

We then decompose $M'$ using SVD,
\[
    U = V^\dag = \begin{bmatrix}1 & 0  \\0 & 1\end{bmatrix},\quad
    \Sigma = \begin{bmatrix}1/\sqrt{2} & 0  \\ 0 &  1/\sqrt{2}\end{bmatrix}.
\]

Because there are just 2 non-zero singular values, we do not need to drop singular values with 2-wide MPS networks and can compute the new MPS as follows: %

\begin{center}
    \begin{tabu}{cc}
        $A_1^{(0)} = [1, 0]$, &
            $A_1^{(1)} = [0, 1]$, \\
        $A_2^{(0)} = [1/\sqrt{2} , 0]^T$, &
            $A_2^{(1)}  = [0, 1/\sqrt{2}]^T$.
    \end{tabu}  
\end{center}

We can see that the output will be $\hat{\rho}=\frac{\ket{00} + \ket{11}}{\sqrt{2}}$ and $\delta = 0$,
since $A_1^{(0)} A_2^{(0)} = A_1^{(1)} A_2^{(1)} = 1/\sqrt{2}$
and other values of $i_0$ and $i_1$ result $0$.

\para{Approximation using 1-wide MPS}
To show how we calculate the approximation error, we use the simplest form of MPS with size $w = 1$, while each $A_i^{(j)}$  becomes a number.

We first initialize the MPS to represent $\ket{00}$:
$$A_1^{(0)} = 1, \ A_1^{(1)} = 0, \ A_2^{(0)} = 1, \ A_2^{(1)} = 0. $$
Then, we apply the $H$ gate to qubit 1:
$$A_1^{(0)} = 1/\sqrt{2}, \ A_1^{(1)} =  1/\sqrt{2}, \ A_2^{(0)} = 1, \ A_2^{(1)} = 0. $$
After that, we apply the CNOT gate and  compute $M$ and $M'$:
\[
    M = \begin{bmatrix}1/\sqrt{2} & 0  \\ 1/\sqrt{2} & 0\end{bmatrix},\quad
    M' = \begin{bmatrix}1/\sqrt{2} & 0  \\ 0 &  1/\sqrt{2}\end{bmatrix}.
\]
We decompose $M'$ using SVD:
\[
    U = V^\dag = \begin{bmatrix}1 & 0  \\0 & 1\end{bmatrix},\quad
    \Sigma = \begin{bmatrix}1/\sqrt{2} & 0  \\ 0 &  1/\sqrt{2}\end{bmatrix}.
\]
Since there are 2 non-zero singular values, we have to drop the lower half with 1-wide MPS tensor networks.
Finally, we obtain $A'_1$ and $A'_2$:
$$A_1^{(0)} = 1, \ A_1^{(1)} =  0, \ A_2^{(0)} = 1/\sqrt{2}, \ A_2^{(1)} = 0. $$
We renormalize the MPS:
$$A_1^{(0)} = 1, \ A_1^{(1)} =  0, \ A_2^{(0)} = 1, \ A_2^{(1)} = 0. $$
Thus, the output approximate state is $\ket{00}$.

To calculate the approximation error bound, we represent the part we drop as an MPS $B$:
$$B_1^{(0)} = 0, \ B_1^{(1)} =  1, \ B_2^{(0)} = 0, \ B_2^{(1)} = \sqrt{2}. $$
Let the unnormalized final state be $\ket{A}$
and the dropped state be $\ket{B}$. Then, the final output is $\sqrt{2}\ket{A}$ and the ideal output is $\ket{A} + \ket{B}$.
The trace distance between the state is
$$\delta = 2 \sqrt{1 - |\braket{\sqrt{2} A | A+B}^2 |} = \sqrt{2} .$$
Therefore, the final output will be $\hat\rho = \ket{00}\bra{00}$ and $\delta = \sqrt{2}$.

\section{Computing the $(\hat{\rho},\delta)$-Diamond Norm}
\label{sec:constrainedSDP}
 Section~\ref{sec:logic} introduces our quantum error logic using the $(\hat{\rho},\delta)$-diamond norm, while treating its 
computation algorithm as a black box.
In this section, we describe how to efficiently calculate the $(\hat{\rho},\delta)$-diamond norm given $(\hat{\rho}, \delta, \mathcal{U}, \mathcal{E})$, \fix{where $\mathcal{U}$ and $\mathcal{E}$ are the perfect and the noisy superoperators respectively.}

\para{Constrained diamond norm}
In $(\hat{\rho},\delta)$-diamond norm, the input state $\rho_{\text{in}}$ is constrained by
\[
    \|\hat\rho - \rho_{\text{in}}\|_1 \le \delta.
\]
We first compute the local density matrix (defined later in this section) $\rho'$ of $\hat\rho$. 
Then, to compute the  $(\hat{\rho},\delta)$-diamond norm, we extend the result of Watrous~\cite{watrous2013simpler} 
by adding the following constraint:
$$\text{tr} (\rho' \rho) \ge \| \rho' \|_F(\| \rho' \|_F - \delta),$$
where  $\rho$ denotes the local density matrix of $\rho_{\text{in}}$.
Thus, $(\hat{\rho},\delta)$-diamond norm can be computed by the following semi-definite program(SDP):

\begin{theorem}
The $(\hat\rho,\delta)$-diamond norm $||\Phi||_{(\hat{\rho}, \delta)}$
can be solved by the semi-definite program(SDP) in Equation~\eqref{eq:new_sdp}.%
\begin{equation}
\begin{array}{ll@{}ll}
\text{maximize}  & \mathrm{tr}(J(\Phi)W) \\
\text{subject to}&  I \otimes \rho \succcurlyeq W  \\
&{ \text{tr} (\rho' \rho) \ge \| \rho' \|_F(\| \rho' \|_F - \delta) }\\
& W \succcurlyeq 0,\ \rho \succcurlyeq 0,\ \mathrm{tr}(\rho) = 1,
\end{array}
\label{eq:new_sdp}
\end{equation}%
\noindent{}where $J$ is the Choi-Jamiolkowski isomorphism \cite{CHOI1975285} and $\Phi = \mathcal{U} - \mathcal{E}.$
\end{theorem}%

\begin{proof}
Given $(\hat\rho, \delta)$, we know that $\|\hat\rho - \rho_{\text{in}}\|_1 \le \delta$, where $\rho_{\text{in}}$ is the real input state.
Let $\rho'$ and $\rho$ be the local density operator of $\hat\rho$ and $\rho_{\text{in}}$. Because partial trace can only decrease the trace norm, we know that
\begin{align*}
    \| \rho' - \rho \|_1 \le \|\hat\rho - \rho_{\text{in}}\|_1 \le \delta.
\end{align*}
For a matrix $\rho$, let $\|\rho \|_F$ be the Frobenius norm which is the square root of the sum of the squares of all elements in a matrix. Because $\|\rho \|_F \le \| \rho \|_1$ for all $\rho$, we know that $ \| \rho' - \rho \|_F \le \delta$.
Then, we have
\begin{align*}
    \text{tr} (\rho' \rho) &= \text{tr}(\rho^2) + \text{tr}((\rho' - \rho)\rho) \\
    &= \| \rho' \|_F^2 +  \text{tr}((\rho' - \rho)\rho) \\
    &\ge \| \rho' \|_F^2 - \|\rho'\|_F \|\rho' - \rho\|_F \\
    &\ge \| \rho' \|_F(\| \rho' \|_F - \delta),
\end{align*}
where the third step holds  because of the Cauchy-Schwarz inequality.

Because the $(Q, \lambda)$-diamond norm can be solved by the following SDP by adding a constraint $ \text{tr} (Q\rho) \ge \lambda$ \cite{hung2019}:
\begin{equation*}
\begin{array}{ll@{}ll}
\text{maximize}  & \mathrm{tr}(J(\Phi)W) \\
\text{subject to}&  I \otimes \rho \succcurlyeq W  \\
&{ \text{tr} (Q \rho) \ge \lambda }\\
& W \succcurlyeq 0,\ \rho \succcurlyeq 0,\ \mathrm{tr}(\rho) = 1.
\end{array}
\end{equation*}
$(\hat\rho, \delta)$-diamond norm can thus be calculated by the SDP \eqref{eq:new_sdp}.
\end{proof}

Let the solved, optimal value of SDP in Equation~\eqref{eq:new_sdp} be $\epsilon$.
We conclude that the $(\hat\rho, \delta)$-diamond norm must be bounded by $\epsilon$, i.e.,
\[
    \|\Phi \|_{ (\hat{\rho}, \delta)} \le \epsilon.
\]

\para{SDP size}
The size of the SDP in Equation~\eqref{eq:new_sdp} is exponential with respect to
the maximum number of quantum gates' input qubits.
Since near-term (NISQ) quantum computers are unlikely to support
quantum gates with greater than two input qubits, 
we can treat the size of the SDP problem as a \emph{constant},
for the purposes of discussing its running time.
\fix{Because the running time of solving an SDP scales polynomially with the size of the SDP, the running time to calculate
$(\hat\rho, \delta)$-diamond norm can be seen as a constant.}

\para{Computing local density matrix}
The local density matrix 
\fix{(also known as reduced density matrix \cite{NielsenChuang})} %
represents the local information of a quantum state.
It is defined using a partial trace on the (global) density 
for the part of the state we want to observe.
For example, the local density operator
on the first qubit of $\frac{\ket{00} + \ket{11}}{\sqrt{2}}$
is \fix{$\frac{1}{2} \begin{psmallmatrix} 1 & 1 \\ 1 & 1 \end{psmallmatrix}$},
meaning that the first qubit of the state is half $\ket 0$ and half $\ket 1$.

In Equation~\eqref{eq:new_sdp}, we need to compute the local density matrix $\rho'$ of $\hat\rho$
about the qubit(s) that the noise represented by $\Phi$ acts on.
$\hat\rho$ is represented by an MPS.
The calculation of a local density operator of an MPS works similarly to how we calculate inner products,
except the wire $i_k$ where $k$ is a qubit that we want to observe.

\section{Evaluation}
\label{sec:evaluation}
\begin{table*}[t]
\centering
\setlength{\tabcolsep}{4pt}
\addtolength{\leftskip} {-2cm} %
    \addtolength{\rightskip}{-2cm}
\begin{tabu}{c|cc|cc|cc|c}
     & Qubit  & Gate  & \framework{} bound & Running &  LQR \cite{hung2019} with  & Running &Worst-case  \\[-3pt]
     Benchmark & number & \fix{count} &  ($\times 10^{-4}$) & time (s) &  full simulator ($\times 10^{-4}$) &time (s)& bound ($\times 10^{-4}$) \\
    \hline
    \texttt{QAOA\_line\_10} & 10 & 27& 0.05 & 2.77 & 0.05 & 215.2 & 27 \\[-3pt]
    \texttt{Isingmodel10} & 10 & 480 & 335.6 & 31.6 & 335.6 & 4701.8 & 480\\[-3pt]
    \texttt{QAOARandom20} & 20 & 160& 136.6 & 19.8 & - &(timed out) &160 \\[-3pt]
    \texttt{QAOA4reg\_20} & 20 & 160 & 138.8 & 12.5 & - &(timed out) &160 \\[-3pt]
    \texttt{QAOA4reg\_30} & 30 & 240 & 207.0 & 25.8 & - &(timed out) &240\\[-3pt]
    \texttt{Isingmodel45} &45 & 2265 & 1739.4 & 338.0 & - &(timed out) &2265\\[-3pt]
    \texttt{QAOA50} &50 & 399 & 344.1 & 58.7 & - &(timed out) &399\\[-3pt]
    \texttt{QAOA75} &75 & 597 & 517.2 & 113.7 & - &(timed out) & 597\\[-3pt]
    \texttt{QAOA100} &100 & 677 & 576.7 & 191.9 & - &(timed out) &677\\
\end{tabu}
\vspace*{0.2cm}
\caption{
    Experimental results of \framework{} ($w=128$) and the baseline on different quantum programs,
    showing the bounds given by \framework{}'s $(\hat\rho, \delta)$-diamond norm,
    the $(Q,\lambda)$-diamond norm with full simulation,
    and \fix{the worst case bound given by} unconstrained diamond norm.
    Experiments time out if they run for longer than 24 hours.
    Note that the worst case bound is directly proportional to the number of gates.
}
\label{tab:evaluation}
\end{table*}

This section evaluates  \framework{} on \change{a set of realistic near-term quantum programs}.
We compare the bounds given by \framework{} to the bounds given by other methods,
as well as the error we experimentally measured from an IBM's real quantum device.
All approximations and full simulations are performed on
an Ubuntu 18.04 server with
an Intel Xeon W-2175 (28 cores @ 4.3 GHz),
62 GB memory, and a 512 GB Intel SSD Pro 600p.

\subsection{Evaluating the computed error bounds}
We evaluated \framework{} on several important quantum programs,
under a sample noise model containing the most common type of quantum noises.
We compared the bounds produced by \framework{}
with the LQR's $(Q, \lambda)$-diamond norm with full simulation 
and the worse-case bounds given by the unconstrained diamond norm.

\para{Noise model}
In our experiments, our quantum circuits are configured such that
each noisy 1-qubit gate has a bit flip ($X$)
with probability $p = 10^{-4}$:
\[
    \Phi(\rho) = (1-p)\rho + p X \rho X.
\]
Each 2-qubit gate also has a bit flip %
on its first qubit.

\para{Framework configuration}
For the approximator, we can adjust the size of the MPS network,
depending on available computational resources;
the larger the size, the tighter error bound.
In \change{all} experiments, we use an MPS of size 128.

\para{Baseline}
To evaluate the error bounds given by \framework{},
we first compared them with the worst-case bounds calculated using the  unconstrained diamond norm (see Section~\ref{sec:qerr}).
For each noisy quantum gate, we  compute its unconstrained diamond norm distance to the perfect gate
and obtain the worst-case bound by summing all unconstrained diamond norms.
The unconstrained diamond norm distance of a bit-flipped gate and a perfect gate is given by:
\begin{align*}
    \|\Phi - I \|_{\diamond} \quad &= \quad \| (pX\circ X + (1-p)I) - I \|_{\diamond} \\
                                   &= \quad p \| X \circ X - I \|_{\diamond} \\
                                   &= \quad p,
\end{align*}
where $X \circ X$ denotes the function that maps $\rho$ to $X \rho X$.
Therefore, the total noise is bounded by $np$,
where $n$ is the number of noisy gates, due to additivity of diamond norms.
Because every gate has a noise, the worst case bound \fix{produced by unconstrained diamond norm} is simply proportional to the number of gates in the program.

We also compare our error bound with what we obtain from LQR~\cite{hung2019}
using a full quantum program simulator to generate best quantum predicate.
This approach's running time is exponential to the number of qubits
and times out (runs for longer than 24 hours) on programs with $\ge 20$ qubits.

\para{Programs}
We analyzed two classes of quantum programs that are expected to be most useful in the near-term,
namely:
\begin{itemize}
  \setlength\itemsep{0.5em}
\item The \emph{Quantum Approximate Optimization Algorithm (QAOA)}~\cite{farhi2014quantum}
    that can be used to solve combinatorial optimization problems.
    We use it to find the max-cut for various graphs with qubit sizes from 10 to 100.

\item The \emph{Ising model}~\cite{google2020hartree}---a thermodynamic model for magnets widely used in quantum mechanics.
    We run the Ising model with sizes 10 and 45.
\end{itemize}

\begin{figure}[t]
    \centering
    \includegraphics[width=.85\columnwidth]{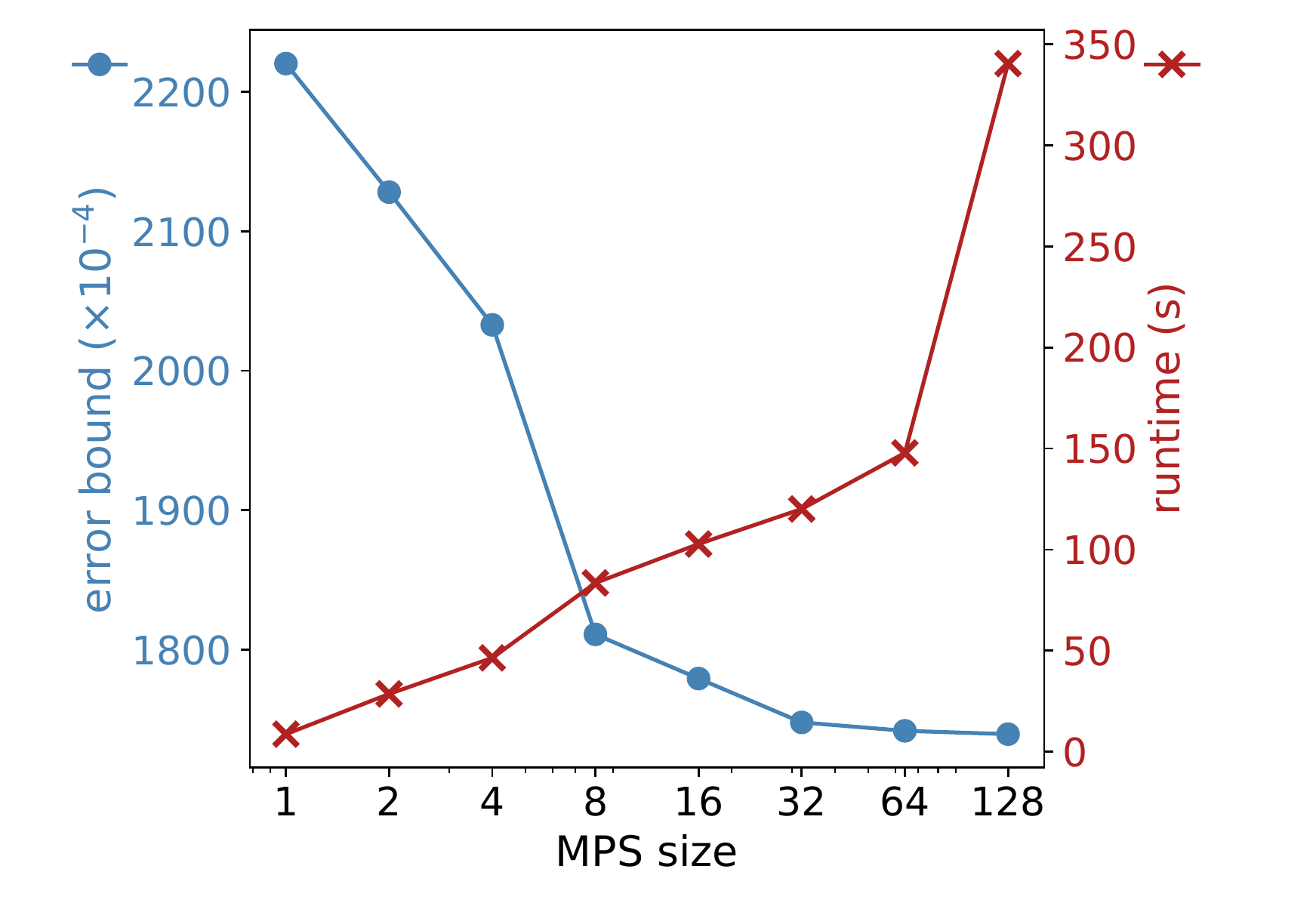}
    \caption{The error bounds and runtimes  of \framework{} on the program \texttt{Isingmodel45} with different MPS sizes.}
    \label{fig:variable-mps-size}
\end{figure}

\para{Evaluation}
\cref{tab:evaluation} presents the evaluation results.
We can see  that 
\framework{}'s bounds are $15\% \sim 30\%$ tighter than 
what the unconstrained diamond norm gives,
on large quantum circuits with qubit sizes $\ge 20$.
On small qubit-size circuits, our bound is as strong as the exponential-time
method based on full simulation.

We also evaluated how MPS size impacts the performance of \framework{}.
As we can see  for the  \texttt{Isingmodel45} program (see  Figure~\ref{fig:variable-mps-size}), larger MPS sizes result in tighter error bounds, at the cost of longer running times,
with marginal returns beyond a certain size.
We found that MPS networks with a size of 128 performed best for our candidate programs,
though in general, this parameter can be adjusted according to
precision requirements and the availability of computational resources.
As the MPS size grows, floating point errors become more significant,
so higher precision representations are necessary for larger MPS sizes.
Note that one cannot feasibly compute the precise error bound
of the \texttt{Isingmodel45} program,
since that requires computing the $2^{45} \times 2^{45}$ matrix representation of the program's output.

\subsection{Evaluating the quantum compilation error mitigation}

\begin{figure}[t]
    \centering
    \includegraphics[width=\columnwidth]{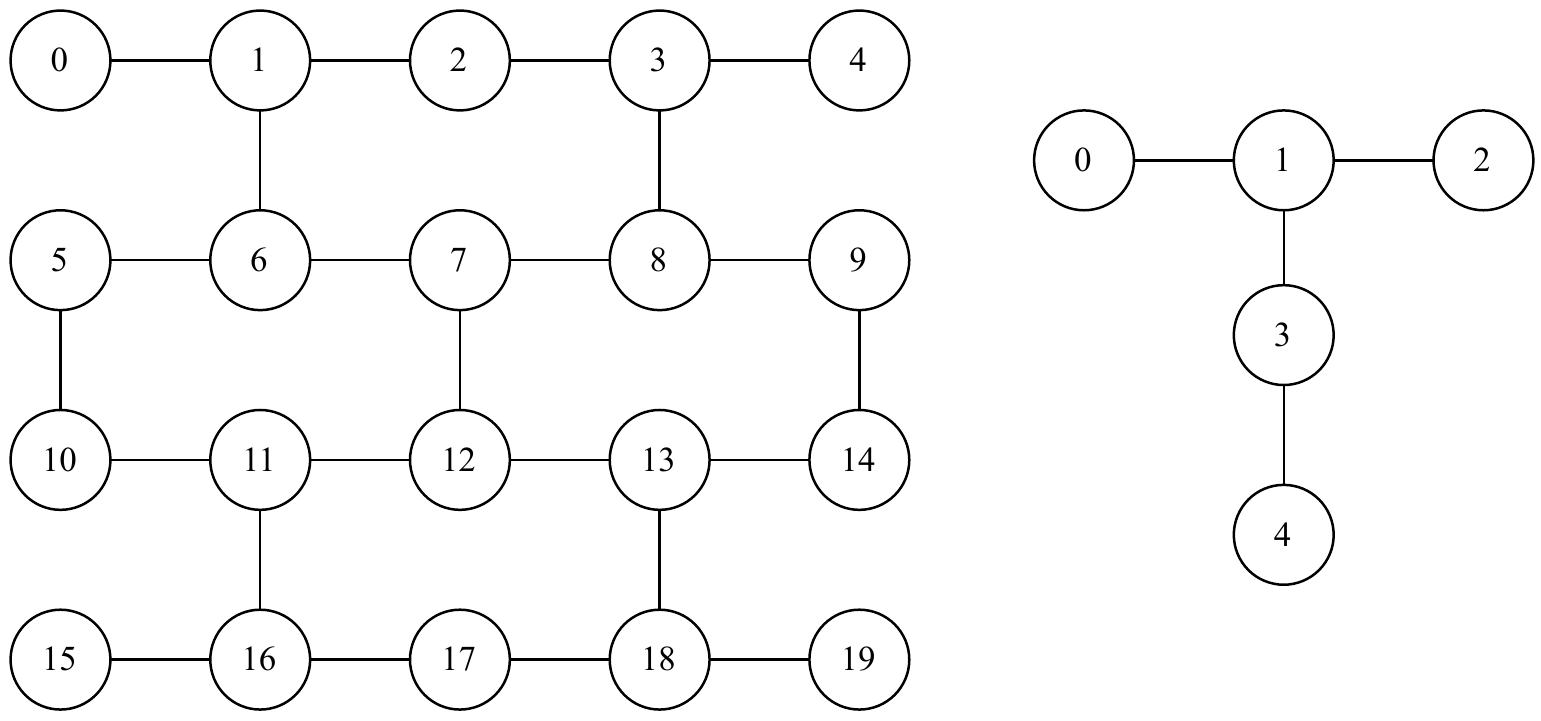}
    \caption{The coupling map of the IBM Boeblingen (left) and Lima (right) quantum computer, where each node represents a qubit.
    Only qubit pairs with a connecting edge can be used to implement a 2-qubit gate.
    }
    \label{fig:device}
\end{figure}

To demonstrate that \framework{} can be used
to evaluate the error mitigation performance of quantum compilers for
real quantum computers today,
we designed an experiment based on the noise-adaptive qubit mapping problem~\cite{8942132, murali2019noiseadaptive}.
When executing a quantum program on a real quantum computer, a quantum compiler must decide
which physical qubit
each logical qubit should be mapped to,
in accordance with the quantum computer's coupling map (e.g., Figure~\ref{fig:device}).
Since quantum devices do not have uniform noise across qubits,
a quantum compiler's mapping protocol should aim to map qubits such that
the quantum program is executed with as little noise as possible.

\begin{figure}
    \centering
\begin{subfigure}[b]{.2\textwidth}
\centering
    \begin{tikzpicture}[scale=0.75]
        \tikzstyle{operator} = [draw,fill=blue!30] 
        \tikzstyle{phase} = [fill,shape=circle,minimum size=5pt,inner sep=0pt]
        \tikzstyle{surround} = [fill=blue!10,thick,draw=black,rounded corners=2mm]
        \tikzset{XOR/.style={draw,circle,append after command={
            [shorten >=\pgflinewidth, shorten <=\pgflinewidth,]
            (\tikzlastnode.north) edge (\tikzlastnode.south)
            (\tikzlastnode.east) edge (\tikzlastnode.west)
            }
            
        }
    }
        \node at (0,0.5) (q1) {$\ket{0}$};
        \node at (0,-0.5) (q2) {$\ket{0}$};
        \node at (0,-1.5) (q3) {$\ket{0}$};
        \node[operator] (op11) at (1,0.5) {$H$} edge [-] (q1);
        \node[phase] (phase11) at (2,0.5) {} edge [-] (op11);
        \node[XOR] (phase12) at (2,-0.5) {} edge [-] (q2);
        \draw[-] (phase11) -- (phase12);
        \node[phase] (phase22) at (3,-0.5) {} edge [-] (phase12);
        \node[XOR] (phase23) at (3,-1.5) {} edge [-] (q3);
        \draw[-] (phase22) -- (phase23);
        \node (end1) at (4,0.5) {} edge [-] (phase11);
        \node (end2) at (4,-0.5) {} edge [-] (phase12);
        \node (end3) at (4,-1.5) {} edge [-] (phase23);
    \end{tikzpicture}
\end{subfigure}\centering
\begin{subfigure}[b]{.2\textwidth}
\centering
    \begin{tikzpicture}[scale=0.75]
        \tikzstyle{operator} = [draw,fill=blue!30] 
        \tikzstyle{phase} = [fill,shape=circle,minimum size=5pt,inner sep=0pt]
        \tikzstyle{surround} = [fill=blue!10,thick,draw=black,rounded corners=2mm]
        \tikzset{XOR/.style={draw,circle,append after command={
            [shorten >=\pgflinewidth, shorten <=\pgflinewidth,]
            (\tikzlastnode.north) edge (\tikzlastnode.south)
            (\tikzlastnode.east) edge (\tikzlastnode.west)
            }
            
        }
    }
        \node at (0,0.5) (q0) {$\ket{0}$};
        \node at (0,0) (q1) {$\ket{0}$};
        \node at (0,-0.5) (q2) {$\ket{0}$};
        \node at (0,-1) (q3) {$\ket{0}$};
        \node at (0,-1.5) (q4) {$\ket{0}$};
        \node[operator] (op11) at (0.7,0.5) {$H$} edge [-] (q0);
        \node[phase] (phase11) at (1.4,0.5) {} edge [-] (op11);
        \node[XOR] (phase12) at (1.4,0) {} edge [-] (q1);
        \draw[-] (phase11) -- (phase12);
        \node[phase] (phase22) at (2.1,0) {} edge [-] (phase12);
        \node[XOR] (phase23) at (2.1,-0.5) {} edge [-] (q2);
        \draw[-] (phase22) -- (phase23);
        \node[phase] (phase31) at (2.8,0) {} edge [-] (phase22);
        \node[XOR] (phase33) at (2.8, -1) {} edge [-] (q3);
        \draw[-] (phase31) -- (phase33);
        
        \node[phase] (phase43) at (3.5,-1) {} edge [-] (phase33);
        \node[XOR] (phase44) at (3.5, -1.5) {} edge [-] (q4);
        \draw[-] (phase43) -- (phase44);
        \node (end1) at (4.2,0.5) {} edge [-] (phase11);
        \node (end2) at (4.2,0) {} edge [-] (phase12);
        \node (end3) at (4.2,-0.5) {} edge [-] (phase23);
        \node (end4) at (4.2,-1) {} edge [-] (phase33);
        \node (end5) at (4.2,-1.5) {} edge [-] (phase44);
    \end{tikzpicture}
\end{subfigure}
    \caption{The GHZ-3 circuit (left) and the GHZ-5 circuit (right).}
    \vspace{10pt}
    \label{fig:ghz_experiment}
\end{figure}
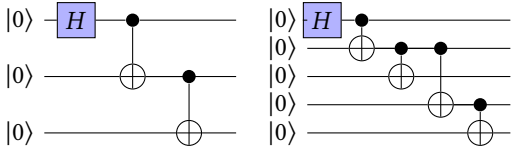

\para{Experiment design}
We compared three different qubit mappings of the
3-qubit GHZ (GHZ-3) circuit and the 5-qubit GHZ circuit (GHZ-5) (see Figure~\ref{fig:ghz_experiment}):
$q_0-q_1-q_2$, $q_1-q_2-q_3$, and $q_2-q_3-q_4$ for GHZ-3, and
$q_0-q_1-q_2-q_3-q_4$ and $q_2-q_1-q_0-q_3-q_4$ for GHZ-5,
where $q_i$ represents the $i$th physical qubit.
As the baseline, we  ran our circuit on a real quantum computer with each qubit mapping
and measured the output to obtain a classical probability distribution.
We computed the measured error
by taking the statistical distance of this distribution from
the distribution of the ideal output state $(\ket{000} + \ket{111})/\sqrt{2}$ and $(\ket{00000} + \ket{11111})/\sqrt{2}$.
We then used \framework{} to compute the noise bound for each mapping,
based on our quantum computer's noise model.
Because the trace distance represents the maximum possible statistical distance
of any measurement on two quantum states (see Section~\ref{sec:qerr}),
the statistical distance we computed should be bounded by
the trace distance computed by \framework{}.

\para{Experiment setup}
We conducted our experiment using the IBM Quantum Experience\cite{IBMQ} platform
and ran  our quantum programs
with the IBM Boeblingen 20-qubit device (see Figure~\ref{fig:device}).
Because \framework{} needs a noise model to compute its error bound,
we constructed a model for the device
using publicly available data from IBM~\cite{IBMQ}
in addition to measurements from tests we ran on the device.

\para{Results}
Our experimental results are shown in \cref{tab:mapping}.
We can see that \framework{}'s bounds are consistent with the real noise levels
and successfully predict the ranking of  noise levels for different mappings.
As for GHZ-3,
the $1-2-3$ mapping has the least noise, while $0-1-2$ has the most.
\framework{}'s bounds are also consistent with the real noise levels
for GHZ-5.
This illustrates how \framework{} can be used to help guide
the design of noise-adaptive mapping protocols---\fix{users can run \framework{} with different mappings and choose the best mapping according to error bounds given by \framework{}.}
In contrast,
the worst case bounds given by the unconstrained diamond norm 
are always $1$ for all five different mappings,
which is not helpful for determining the best mapping.

\begin{table}[t]
    \centering
    \setlength{\tabcolsep}{4pt}
\addtolength{\leftskip} {-2cm} %
    \addtolength{\rightskip}{-2cm}
    \begin{tabular}{c|c|cc}
        Circuit & Mapping & \framework{} bound & Measured error\\
        \hline
        GHZ-3 & 0-1-2 & 0.211  & 0.160\\[2pt]
        GHZ-3 & 1-2-3 & 0.128  & 0.073\\[2pt]
        GHZ-3 & 2-3-4 & 0.162  & 0.092\\[2pt]
        \hline
        GHZ-5 & 0-1-2-3-4 & 0.471 & 0.176\\[2pt]
        GHZ-5 & 2-1-0-3-4 & 0.449 & 0.171\\
    \end{tabular}
    \vspace{5pt}
    \caption{
        Error bounds generated by \framework{} on different mappings
        compared with the errors we measured experimentally
        using the IBM  Boeblingen 20-qubit device.
    }
    \label{tab:mapping}
\end{table}
\section{Related Work}

\para{Error bounding quantum programs}
Robust projective quantum Hoare logic~\cite{zhou2019} is an extension of
Quantum Hoare Logic that supports error bounding using the worst-case diamond norm.
In contrast, \framework{} uses the more fine-grained $(\hat\rho, \delta)-$diamond norm
to provide tighter %
error bounding.

LQR~\cite{hung2019} is a framework for
formally reasoning
about quantum program errors, 
using the $(Q,\lambda)$-diamond norm as its error metric.
LQR supports the reasoning about quantum programs 
 with more advanced
quantum computing features, such as quantum loops.
However, LQR does not specify any practical method for obtaining non-trivial quantum predicates.
{In contrast, \framework{}, for the first time,
introduces a practical and adaptive method to compute quantum program predicates, i.e., $(\hat{\rho}, \delta)$ predicates, using the $TN$ algorithm.
}

As we have shown in Section~\ref{sec:constrainedSDP}, our  $(\hat{\rho}, \delta)$ predicates
can be reduced to LQR's $(Q, \lambda)$ predicates.
In other words, our quantum error logic
can be understood as a refined implementation of LQR.
$(\hat{\rho}, \delta)$ predicates computed using \framework{}
can be used to obtain non-trivial postconditions
for the quantum Hoare triples required by LQR's sequence rule.
By the soundness of our $TN$ algorithm,
the computed predicates are guaranteed to be valid postconditions.

\para{Error simulation}
Contemporary error simulation methods can be roughly divided into two classes:
(1) direct simulation methods
based on solving Schr\"{o}dinger's equation or the master equation~\cite{markov1995}---%
neither of which scales beyond a few qubits~\cite{Pashayan2017FromEO}---and
(2) approximate methods,
based on either Clifford circuit approximation~\cite{Guti_rrez_2013,Magesan_2013,krbrown_2016, Bravyi_2019}
or classical sampling methods with Monte-Carlo
simulations~\cite{Rauendorf2019PhaseSS,Veitch2012NegativeQA,Mari_2012,Veitch2013EfficientSS}.
These methods are efficient but only work on specific classes 
of quantum circuits such as circuits and noises represented by positive Wigner functions or Clifford gates.
In contrast, \framework{} can be applied to general quantum circuits
and scales well beyond 20 qubits.

\para{Resource estimation beyond error}
Quantum compilers such as Qiskit Terra~\cite{Qiskit} and ScaffCC~\cite{JavadiAbhari2015ScaffCCSC}
perform entanglement analysis for quantum programs.
The QuRE~\cite{qure} toolbox provides coarse-grained resource estimation
for fault-tolerant implementations of quantum algorithms.
On the theoretical side, quantum resource theories also consider the estimation of 
coherence~\cite{Streltsov2017ColloquiumQC,Winter2016OperationalRT},
entanglement~\cite{plenio2005introduction,PhysRevA.74.012305},
and magic state stability~\cite{Howard_2017,wang2019quantifying,Veitch_2014}.
However, these frameworks directly use the matrix representation of quantum states and
do not work for %
quantum programs with more than 20 qubits
that can be  handled by \framework{}.

\para{Verification of quantum compilation}
CertiQ~\cite{certiq} is an automated framework to  verify
that the quantum compiler passes and optimizations 
preseve the semantics of quantum circuits. 
VOQC~\cite{voqc} is a formally verified optimizer for quantum circuits.
These works focus on quantum compilation correctness
and do not consider noise models or error-mitigation performance.
In contrast, 
Gleipnir focuses on the error anaylysis of quantum programs
and can be used to evaluate the error-mitigation performance in quantum compilations.

\para{Tensor network quantum approximation}
\fix{MPS and general tensor networks are mostly used in the exact evaluation of quantum programs. 
The only application of MPS for approximating quantum systems is the density matrix renormalization group (DMRG) method in quantum chemistry \cite{r2}. Although both DMRG and Gleipnir use MPS to represent approximate quantum states, the approximation methods are different.
DMRG can only be used to simulate quantum many-body systems,
while Gleipnir’s approach works for general programs and can provide the error bounds of the approximate states, which are used in the quantum error logic to compute the error bounds of quantum programs.}

Multi-dimensional tensor networks such as
PEPS~\cite{jordan2008classical} and MERA~\cite{giovannetti2008quantum}
may model quantum states more precisely than MPS. %
However, they are computationally impractical.
Contracting higher-dimensional tensor networks involves tensors with orders greater than four,
which are prohibitively expensive to manipulate.

\section{Conclusion}
We have presented \framework, a methodology for computing
verified error bounds of quantum programs
and evaluating the error mitigation performance of quantum compiler transformations.
Our experimental results show that \framework{} provides %
tighter error bounds
for quantum circuits with qubits ranging from 10 to 100, \fix{compared with the worst case bound},
and the generated error bounds are consistent with the noise-levels
measured using real quantum devices.

\section{Acknowledgement}
We thank our shepherd, Timon Gehr, and the anonymous reviewers for valuable feedbacks that help improving this paper.
We thank Xupeng Li and Shaokai Lin for conducting parts of the experiments.
We thank members of the VeriGu Lab at Columbia and anonymous
referees for helpful comments and suggestions that improved this
paper and the implemented tools.
This work is funded in part by NSF grants CCF-1918400 and CNS-2052947;
an Amazon Research Award;
EPiQC, an NSF Expedition in Computing,
under grants CCF-1730449; STAQ under grant
NSF Phy-1818914;  DOE grants DE-SC0020289 and DESC0020331;
and NSF OMA-2016136 and the Q-NEXT
DOE NQI Center.
Any opinions, findings, conclusions, or recommendations that
are expressed herein are those of the authors, and do not
necessarily reflect those of the US Government, NSF, DOE, 
or Amazon.

\bibliography{reference}
 \begin{appendices}
 \section{Soundness Proof of Quantum Error Logic}
\label{subsec:soundness}

In this section, we prove that our quantum error logic can be used to soundly reason
about the error bounds of quantum programs. For clarity, we omit the parentheses when applying the denotational semantics to a state, i.e., 
$$[\![P]\!]\rho := [\![P]\!](\rho).$$

\begin{theorem}[Soundness of the quantum error logic]
If $(\hat{\rho}, \delta)\vdash\widetilde{P}_\noise\le\epsilon$,
then for all density matrices $\rho$ such that
$\|\rho - \hat{\rho}\|_1 \le \delta$:
\[
    \frac{1}{2}\big\|[\![\widetilde{P}_\noise]\!]\rho - [\![P]\!]\rho\big\|_1 \le \epsilon
\]
\end{theorem}

\begin{proof}

We prove this statement by induction over the judgement $(\hat{\rho}, \delta)\vdash\widetilde{P}_\noise\le\epsilon$,
considering each of the five inference rules in Section~\ref{sec:logic}
that may have been used to logically derive the error bound.

\para{\textsc{Skip} rule}
Because the noisy and perfect semantics of $\texttt{skip}$ are both identity functions,
we know that, trivially,
$[\![\widetilde{P}_\noise]\!]\rho = [\![P]\!]\rho = \rho$
and $\frac{1}{2}\big\|[\![\widetilde{P}_\noise]\!]\rho,\ [\![P]\!]\rho\big\|_1 = 0$.

\para{\textsc{Weaken} rule}
By induction, we know that,
for some $\delta' \ge \delta$ and $\epsilon' \le \epsilon$:
\[
    \|\rho - \hat\rho\|_1 \le \delta'
    \quad \Rightarrow \quad
    \frac{1}{2}\big\|[\![\widetilde{P}_\noise]\!]\rho - [\![P]\!]\rho\big\|_1
        \le \epsilon' 
\]
If, for some fixed $\rho$, we have $ \mathrm{T}(\rho, \hat\rho) \le \delta$,
by $\delta' \ge \delta$, we know that $ \mathrm{T}(\rho, \hat\rho) \le \delta'$, which
implies $\mathrm{T}\big([\![\widetilde{P}_\noise]\!]\rho,\ [\![P]\!]\hat\rho \big) \le \epsilon' \le \epsilon$.
Thus:
\[
    \|\rho - \hat\rho\|_1 \le \delta
    \quad \Rightarrow \quad
    \frac{1}{2}\big\|[\![\widetilde{P}_\noise]\!]\rho - [\![P]\!]\rho\big\|_1 \le \epsilon.
\]

\para{\textsc{Gate} rule}
If the error bound was derived using the \textsc{Gate} rule,
we may use the definition of the $(\hat\rho,\delta)$-diamond norm defined in Section~\ref{sec:logic}
to show that:
{\small
  \setlength{\abovedisplayskip}{6pt}
  \setlength{\belowdisplayskip}{\abovedisplayskip}
  \setlength{\abovedisplayshortskip}{0pt}
  \setlength{\belowdisplayshortskip}{3pt}
\begin{align*}
&\frac{1}{2}\Big\|[\![\widetilde{U}_\noise(q_1, \ldots, q_k)]\!]\rho - [\![U(q_1, \ldots, q_k)]\!]\rho\Big\|_1 \\
    =&\      \frac{1}{2} \Big\|\widetilde{\mathcal{U}}_\noise(\rho) - \mathcal{U}(\rho)\Big\|_1 \\
    \le&\      \|\widetilde{\mathcal{U}}_\noise-\mathcal{U}\|_{(\hat\rho, \delta)} \\
    \le&  \   \epsilon 
\end{align*}
}

\para{\textsc{Seq} rule}
Let $\widetilde{P}_\noise := \widetilde{P}_{1\noise}; \widetilde{P}_{2\noise}$.
By the \textsc{Seq} rule, we have the following induction hypotheses:
\begin{eqnarray}
    \|\rho - \hat\rho_1\|_1 \le \delta_1
     &\Rightarrow \ 
    \frac{1}{2}\big\|[\![\widetilde{P}_{1\noise}]\!]\rho - [\![P_1]\!]\rho\big\|_1
        \le \epsilon_1 \nonumber
\\
    \|\rho - \hat\rho_2\|_1 \le \delta_1 + \delta_2
   &\Rightarrow \ 
    \frac{1}{2}\big\|[\![\widetilde{P}_{2\noise}]\!]\rho - [\![P_2]\!]\rho\big\|_1
        \le \epsilon_2
        \label{eq:delta2}
\end{eqnarray}
To prove $\frac{1}{2}\big\|[\![ \widetilde{P}_{1\noise}; \widetilde{P}_{2\noise}]\!]\rho - [\![P_1; P_2]\!]\rho\big\|_1  \leq \epsilon_1 + \epsilon_2$, we start from the left formula:
{\small
  \setlength{\abovedisplayskip}{6pt}
  \setlength{\belowdisplayskip}{\abovedisplayskip}
  \setlength{\abovedisplayshortskip}{0pt}
  \setlength{\belowdisplayshortskip}{3pt}
\begin{align}
 &\hspace{-7pt} \frac{1}{2}\|[\![ \widetilde{P}_{1\noise}; \widetilde{P}_{2\noise}]\!]\rho - [\![P_1; P_2]\!]\rho\|_1\nonumber \\
 = &     \ \nonumber
  \frac{1}{2}  \|[\![\widetilde{P}_{2\noise}]\!][\![\widetilde{P}_{1\noise}]\!]\rho-[\![P_2]\!][\![P_1]\!]\rho\|_1  \nonumber  \\
 \le&   \ \nonumber
   \frac{1}{2} \|[\![\widetilde{P}_{2\noise}]\!][\![\widetilde{P}_{1\noise}]\!]\rho-[\![\widetilde{P}_{2\noise}]\!][\![P_1]\!]\rho\|_1  \nonumber  \\
    & + 
    \frac{1}{2}\|[\![\widetilde{P}_{2\noise}]\!][\![P_1]\!]\rho-[\![P_2]\!][\![P_1]\!]\rho\|_1  \nonumber  \\
 \le&   \
    \frac{1}{2}\|[\![\widetilde{P}_{1\noise}]\!]\rho-[\![P_1]\!]\rho\|_1  \nonumber  \\
    &+
    \frac{1}{2}\|[\![\widetilde{P}_{2\noise}]\!][\![P_1]\!]\rho-[\![P_2]\!][\![P_1]\!]\rho\|_1
\label{eq:seq_rule}
\end{align}
}%
The last step holds because superoperators ($[\![\widetilde{P}_{2\noise}]\!]$ in our case)
are trace-contracting, and thus trace distance-contracting.
Thus, the error bound is bounded by the sum of two terms.
The first term is bound by $\epsilon_1$
since we have $(\hat\rho, \delta_1)\vdash\widetilde{P}_{1\noise}\le\epsilon_1$ and the induction hypothesis.
From the result of $\text{TN}(\hat\rho_1, P_1) = (\hat\rho_2, \delta_2)$,
we know that $T(\hat\rho_2,[\![P_1]\!]\hat\rho_1) \le \delta_2 $, then we have:
\begin{align*}
    \|\hat\rho_2 - [\![P_1]\!]\rho\|_1
    &\le    \|\hat\rho_2 - [\![P_1]\!]\hat\rho_1\|_1 + \|[\![P_1]\!]\hat\rho_1 - [\![P_1]\!]\rho\|_1 \\
    &\le    \delta_2 + \|\hat\rho_1 - \rho\|_1\\
    &\le    \delta_2 +  \delta_1 %
\end{align*}
Substituting $[\![P_1]\!]\rho$ in place of $\rho$ in Equation~\eqref{eq:delta2} implies:
\[
\frac{1}{2}\|[\![\widetilde{P}_{2\noise}]\!][\![P_1]\!]\rho - [\![P_2]\!][\![P_1]\!]\rho\|_1 \le \epsilon_2
\]
Hence, Equation~\eqref{eq:seq_rule} is bound by $\epsilon_1 + \epsilon_2$,
proving the soundness of the \textsc{Seq} rule.

\para{\textsc{Meas} rule}
Let $\widetilde{P}_\noise := \texttt{if } q = \ket{0} \texttt{then } \widetilde{P}_{0\noise} \texttt{ else } \widetilde{P}_{1\noise}$.
When measuring $\hat\rho$, the result is $0$ with probability $p_0$ and collapsed state $\hat\rho_0$,
and $1$ with probability $p_1$ and collapsed state $\hat\rho_1$;
when measuring $\rho$, the collapsed state is $\rho_{0}, \rho_{1}$.
Thus, we have:
{\small
  \setlength{\abovedisplayskip}{6pt}
  \setlength{\belowdisplayskip}{\abovedisplayskip}
  \setlength{\abovedisplayshortskip}{0pt}
  \setlength{\belowdisplayshortskip}{3pt}
\begin{align*}
     \frac{1}{2}\|[\![\widetilde{P}_\noise]\!]\rho - [\![P]\!]\rho\|_1  
    \quad\le\quad &(1-\delta)\Big(\frac{1}{2}p_0 \|[\![\widetilde{P}_{0\noise}]\!]\rho_{0}- [\![P_0]\!]\rho_0\|_1 \\
                  &+\frac{1}{2} p_1 \|[\![\widetilde{P}_{1\noise}]\!]\rho_{1}- [\![P_1]\!]\rho_1\|_1 \Big) + \delta \\
    \le \quad &(1-\delta) \epsilon + \delta
\end{align*}
}%
The first step means that the probability that the measurement result of $\rho$ and $\rho$
are different is at most $\delta$ by the definition of trace distance.
If the measurement results are the same, then the corresponding branch should be executed
and the noise level is the weighted average of both branches.
In each branch, the preconditions $\|\hat\rho_0 - \rho_0\|_1 \le \delta$ and $\|\hat\rho_1 - \rho_1\|_1 \le \delta$ still hold
because projection does not increase trace distance.
By our induction hypotheses,
$(\hat{\rho}_0, \delta)\vdash\widetilde{P}_{0\noise}\leq\epsilon$
and $(\hat{\rho}_1, \delta)\vdash\widetilde{P}_{1\noise}\leq\epsilon$
the error in each branch is bound by $\epsilon$,
and thus the total error is at most $(1-\delta) \epsilon+ \delta$.
\end{proof}

 \end{appendices}

\end{document}